\documentclass[10pt,journal,compsoc]{IEEEtran}
\IEEEoverridecommandlockouts
\usepackage[linesnumbered, vlined, ruled]{algorithm2e}
\usepackage{mathrsfs}
\usepackage{bm}
\usepackage[usenames, dvipsnames]{color}
\usepackage{colortbl}
\usepackage{tabularx}
\usepackage{booktabs}
\usepackage{threeparttable}
\usepackage{graphicx}
\usepackage{subfigure}
\usepackage{hhline}
\usepackage[fleqn]{amsmath}
\usepackage{amsfonts,amssymb}
\usepackage{amsthm}
\usepackage[noend]{algpseudocode}
\usepackage{longtable}

\newtheorem{theorem}{Theorem}
\newtheorem{lemma}{Lemma}

\newtheorem{definition}{Definition}

\definecolor{BrickRed}{RGB}{178,34,34}

\makeatletter  
\def\ps@IEEEtitlepagestyle{%
  \def\@oddfoot{\mycopyrightnotice}%
  \def\@evenfoot{}%
}
\def\mycopyrightnotice{%
  {\footnotesize 
  \begin{minipage}{\textwidth}
  \centering
  \textbf{This work has been submitted to the IEEE for possible publication. Copyright may be transferred without notice, after which this version may no longer be accessible.}
  \end{minipage}
    \hfill}
  \gdef\mycopyrightnotice{}
}

\usepackage[
	n,
	operators,
	advantage,
	sets,
	adversary,
	landau,
	probability,
	notions,	
	ff,
	mm,
	primitives,
	events,
	complexity,
	asymptotics,
	keys]{cryptocode}
	
\begin{document}

\title{SPDL: Blockchain-secured and Privacy-preserving Decentralized Learning}

\author{Minghui~Xu,~\IEEEmembership{Member,~IEEE,}
    Zongrui~Zou,
    Ye~Cheng, 
    Qin~Hu,~\IEEEmembership{Member,~IEEE,}
    Dongxiao~Yu,~\IEEEmembership{Senior~Member,~IEEE,}
    Xiuzhen~Cheng,~\IEEEmembership{Fellow,~IEEE}
\thanks{M. Xu, Z. Zou, Y. Cheng, D. Yu, and X. Cheng are with the School of Computer Science and Technology, Shandong University, Qingdao, 266510, P. R. China. E-mail: mhxu@sdu.edu.cn; zou.zongrui@mail.sdu.edu.cn; chengye0311@163.com; \{xzcheng, dxyu\}@sdu.edu.cn}
\thanks{Q. Hu is with the Department of Computer and Information Science, Indiana University-Purdue University Indianapolis, USA. E-mail: qinhu@iu.edu}}


\markboth{Journal of \LaTeX\ Class Files,~Vol.~14, No.~8, August~2015}%
{Shell \MakeLowercase{\textit{et al.}}: Bare Demo of IEEEtran.cls for IEEE Journals}

\IEEEtitleabstractindextext{
\begin{abstract}
	Decentralized learning involves training machine learning models over remote mobile devices, edge servers, or cloud servers while keeping data localized. Even though many studies have shown the feasibility of preserving privacy, enhancing training performance or introducing Byzantine resilience, but none of them simultaneously considers all of them. Therefore we face the following problem: \textit{how can we efficiently coordinate the decentralized learning process while simultaneously maintaining learning security and data privacy?} To address this issue, in this paper we propose SPDL, a blockchain-secured and privacy-preserving decentralized learning scheme. SPDL integrates blockchain, Byzantine Fault-Tolerant (BFT) consensus, BFT Gradients Aggregation Rule (GAR), and differential privacy seamlessly into one system, ensuring efficient machine learning while maintaining data privacy, Byzantine fault tolerance, transparency, and traceability. To validate our scheme, we provide rigorous analysis on convergence and regret in the presence of Byzantine nodes. We also  build a SPDL prototype and conduct extensive experiments to demonstrate that SPDL is effective and efficient with strong security and privacy guarantees.
\end{abstract}

}

\maketitle

\IEEEpeerreviewmaketitle

\section{Introduction}


With the increasing amount of data and growing complexity of machine learning models, there is a rigid demand for utilizing computational hardware and storage owned by various entities in a distributed network. State-of-the-art distributed machine learning schemes adopt three major network topologies shown in Fig.~\ref{fig:distributed}. Federated learning \cite{bonawitz2017practical} utilizes an efficient centralized network illustrated in Fig.~\ref{fig:distributed}(a), where a parameter server aggregates gradients computed by distributed devices and updates the global model for them while preserving privacy since devices compute locally without communicating with each other. The fragility of the centralized network topology lies in that a centralized parameter server suffers from the single point of failure problem (a server might crash or be Byzantine). 
To solve this issue, El-Mhamdi \textit{et al.} \cite{el2020genuinely} proposed a Byzantine-resilient learning network shown in Fig.~\ref{fig:distributed}(b), which substitutes the centralized server with a server group in which no more than $1/3$ servers can be Byzantine. 

\begin{figure}[!htbp]
	\centering
	\includegraphics[width=3.5in]{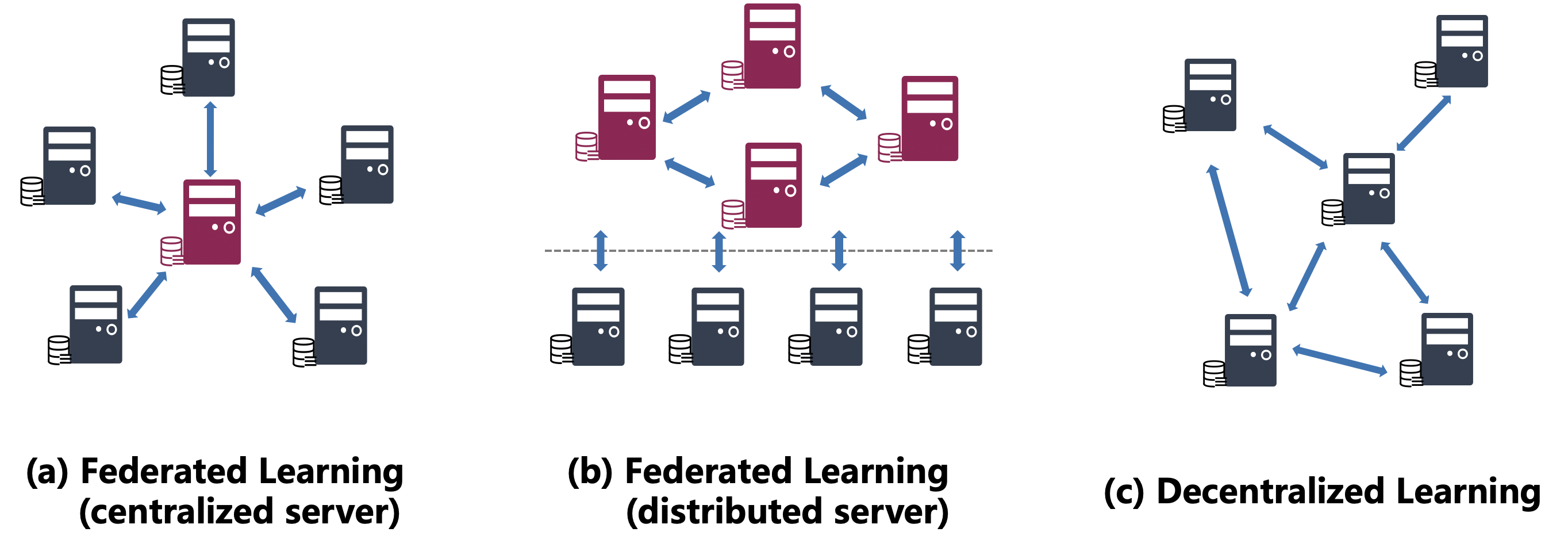}
	\caption{Three major network topologies adopted in distributed learning}
	\label{fig:distributed}
\end{figure}

In this paper, we make a step further to break the barriers among the parameter servers and the computational devices, and allow all nodes to train models in a decentralized network, as shown in Fig.~\ref{fig:distributed}(3). Such a decentralized network is frequently adopted in ad hoc networks, edge computing, Internet-of-Things (IoT), decentralized applications (Dapp), etc. It can greatly unleash the potential for building large-scale (even worldwide) machine learning models that can reasonably and fully maximize the utilization of computational resources \cite{lian2017can}. Besides, devices such as mobile phones, IoT sensors and vehicles are generating a large amount of data nowadays. A decentralized network filled with real-time big data can lead to tremendous improvements in large-scale applications, e.g., illness detection, outbreak discovery, and disaster warning, which involve a large number of decentralized edge and cloud servers from different regions and/or countries. 
However, this poses a new critical challenge:
\textit{How can we efficiently coordinate the decentralized learning process while simultaneously maintaining learning security and data privacy?}



%
Specifically, decentralized learning confronts with the following challenges.
1) Without a fully trusted centralized custodian, users have no incentive but is reluctant to participate in the learning process due to the lack of trust in a decentralized network. Therefore, it is highly possible that the volume of data might be insufficient to train a reliable model.
2) In decentralized learning, a Byzantine node who can behave arbitrarily (e.g., crash or launch attacks) might prevent the model from convergence or interrupt the training process. 
3) It is challenging to make the trade-off between privacy \& security as well as efficiency, and to make data sharing frictionless with privacy and transparency guarantees. 

To overcome the above challenges, we propose SPDL, a decentralized learning framework that simultaneously ensures strong security using blockchain (as an immutable distributed ledger), BFT consensus, and BFT GAR, and  preserves privacy utilizing local gradient computation and differential privacy (DP). Blockchain, as a key component, maintains a complete, immutable, traceable record of the machine learning process, covering user registration, gradients by rounds, and model parameters. With blockchain, a user can identify illegal and Byzantine peers, update its local model without concerns, and ultimately trust the SPDL scheme. The BFT consensus algorithm and the BFT GAR are embedded into the blockchain. Concretely, the BFT consensus algorithm ensures consistency of model transition during multiple rounds while the BFT GAR offers an effective method of detecting and filtering Byzantine gradients at each round. Concerning privacy, we let nodes compute gradients with their local training data, and only share perturbed gradients with peers, which provides a strong privacy protection. 

Our contributions are summarized as follows
\begin{enumerate}
\item To our best knowledge, this is the first secure and privacy-preserving machine learning scheme for decentralized networks in which learning processes are free from trusted parameter servers. 
\item SPDL makes use of DP for data privacy protection and seamlessly embeds BFT consensus and BFT GAR into a blockchain system to benefit model training with Byzantine fault tolerance, transparency, and traceability while retaining high efficiency. 
\item We conduct rigorous convergence and regret analyses on SPDL in the presence of Byzantine nodes,  build a prototype, and carry out extensive experiments to demonstrate the feasibility and effectiveness of SPDL.
\end{enumerate}

This paper is organized as follows. Section~\ref{sec:preliminaries} outlines the necessary preliminary knowledge needed by the development of SPDL. Section~\ref{sec:protocol} details the SPDL protocol. The analysis on convergence and regret are presented in Section~\ref{sec:analysis}. Evaluation results are reported in Section~\ref{sec:evaluation}. We summarize the most related work in Section~\ref{sec:related} and conclude this paper in Section~\ref{sec:conclusion}.

\section{Related Work} \label{sec:related}
\subsection{Privacy and Byzantine Resilience in Distributed Learning}
Private learning schemes include secure multiparty computation, encryption,  homomorphic encryption,  differential privacy,  and aggregation models.  
For details, we recommend two comprehensive surveys \cite{liu2021machine, yang2019federated}  to the interested readers.  Su and Vaidya  \cite{su2016fault} introduced the distributed optimization problem in the presence of Byzantine failures. The problem was formulated as one in which each node has a local cost function, and aims to optimize the global cost function.  The proposed method, namely the synchronous Byzantine gradient method (SBG), first trims the largest $f$ gradients and the smallest $f$ gradients, then computes the average of the minimum and the maximum of the remaining $N-2f$ values. This approach sheds light on providing byzantine resilience for distributed learning. Following this idea, many byzantine-resilient aggregation rules were proposed. They all work towards a common goal -- more precisely and efficiently trim the byzantine values. 
Blanchard  \textit{et al.} were the earliest to tackle the Byzantine resilience problem in distributed learning by the Krum algorithm which can guarantee convergence despite $f$ Byzantine workers (in a server-worker architecture).  Krum stimulates the arrivals of many BFT aggregation rules including Median \cite{yin2018byzantine},  Bulyan \cite{guerraoui2018hidden}, and  MDA \cite{el2020genuinely}.  
A recent work~\cite{guerraoui2021differential} demonstrates that these aggregation rules can function together with the DP technique under proper assumptions.

\subsection{Blockchain-Enhanced Distributed Learning}
%
%
The BinDaaS \cite{bhattacharya2019bindaas} framework provides a blockchain-based deep learning service, ensuring data privacy and confidentiality in sharing Electronic Health Records (EHRs). With BinDaaS, a patient can mine a block filled with a private health record, which can be accessed by legitimate doctors. BinDaaS trains each model based on a patient's private EHRs independent of others' data. In contrast, decentralized learning intends to train a global model using the data owned by individual nodes, arising more security and privacy concerns. 
Hu \textit{et al.} \cite{9615370} utilized a blockchain and a game-theoretic approach to protect the user privacy for federated learning in mobile crowdsensing. The collective extortion (CE) strategy was proposed in \cite{hu2021nothing} as an incentive mechanism that can regulate workers's behavior. These two game-based approaches cannot strictly guarantee the safety of model training against byzantine nodes.  
Lu \textit{et al.} \cite{lu2019blockchain} developed a learning scheme that protects privacy by differential privacy, and proposed the Proof of Training Quality (PoQ) consensus algorithm for model convergence. This scheme does not consider byzantine users who might disturb the training processes by proposing erroneous model parameters.

FL-Block \cite{FLblock} allows end devices to train a global model secured by a PoW-based blockchain. LearningChain \cite{learning2018chain} is a differential privacy based scheme to protect each party’s data privacy, also resting on a  PoW-based blockchain.  Warnat-Herresthal \textit{et al.} \cite{warnat2021swarm} proposed the concept of swarm learning which is analogous to decentralized learning, in which each node can join the learning process managed by Ethereum (again, PoW-based) and smart contracts. 
Compared to a pure adoption of centralized federated learning schemes, PoW-based blockchains can help avoid the single point of failures and mitigate poisoning attacks caused by a central parameter server. With PoW, a miner can propose a block containing model parameters used for a global model update. However, this method cannot rigorously prevent malicious miners from harming the training processes by proposing wrong updates. The PoW consensus itself is not sufficient to judge whether given parameters are byzantine or not. Therefore, PoW-based blockchains are vulnerable to poisoning attacks launched by byzantine nodes. Besides, PoW-based blockchains are hard to scale since PoW can incur much overhead and result in heavy consumption of computational resources. 
Another noteworthy system is Biscotti \cite{shayan2020biscotti}, which utilizes a commitment scheme to protect data privacy and adopts a novel Proof-of-Federation based blockchain for security guarantee in decentralized learning.  Using a commitment scheme, the aggregation of parameters can be verifiable and tamper-proof. 

In this paper, we leverage the DP technique for privacy protection since adding noises (only needs an addition operation using pre-calculated noise) is more efficient than employing complicated cryptographic tools. In addition, using DP and a BFT consensus, we can ensure the byzantine fault tolerance for the whole training process, which can not be realized by existing works. Our unique contributions lie in two aspects: 1) providing security and privacy guarantees for the complete training process, by seamlessly integrating DP, BFT aggregation, and blockchain in one system,  considering byzantine nodes; and 2) offering rigorous analysis on the Byzantine fault-tolerance, convergence and regret for decentralized learning. 
%


\section{Preliminaries} \label{sec:preliminaries}

\subsection{Decentralized Learning}
\label{sec:sub:decentralized:learning}
With the prosperity of machine learning tasks that need to process massive data,  distributed learning was proposed to coordinate a large number of devices 
to complete a training process so as to achieve a rapid convergence.  In recent years, the research on distributed learning was mainly carried out along two directions: centralized learning and decentralized learning, whose schematic diagrams are shown in Fig.~\ref{fig:distributed}.  A centralized topology uses a parameter server (PS) to coordinate all workers by gathering gradients,  performing local updates, and broadcasting new parameters; while in a decentralized topology,  all nodes are considered as equal and exchanges information without the intervention of a PS.

Decentralized learning has many advantages over centralized learning.  In \cite{lian2017can},  Lian \textit{et al.} rigorously proved that a decentralized algorithm has lower communication complexity and possesses the same convergence rate as those under a centralized parameter-server model. 
Furthermore, a centralized topology might not hold in a decentralized network where no one can be trusted enough to act as a parameter server. Therefore, We consider the following decentralized optimization during a learning process:
$$\min\limits_{x\in\mathbb{R}^N}   f(x) = \frac{1}{N}\sum_{i=1}^{N}\mathbb{E}_{\xi\sim D_i} F(x;\xi),$$
where $N$ is the network size, $D_i$ is the local data distribution for node $i$, $F(x;\xi)$ denotes the  loss function given model parameter $x$ and data sample $\xi$. Let $f_{i}(x)=\mathbb{E}_{\xi\sim D_i} F(x;\xi)$.  
In a fully-connected graph, during any synchronous round, each node performs a deterministic aggregation function (\textit{e.g.} average function) $\mathcal{K}$ on perturbed gradients received from all other peers to update its local parameter, i.e.,
$$x_i^{(t+1)} = x_i^{(t)} - \gamma \cdot \mathcal{K}(g_1^{(t)},g_2^{(t)},\cdots,g_n^{(t)}),$$
where $\gamma$ is the learning rate. Notice that since $\mathcal{K}$ is deterministic, all nodes should have exactly the same parameter and should initialize it with the same value. In this case, we denote by $\mathcal{A}$ an arbitrary synchronous distributed learning algorithm that updates the mutual parameter $x$, and use $\mathcal{A}(D_1,D_2,\cdots,D_n;t)$ to denote the output parameter $x$ at round $t$. In the rest of this paper we omit the subscript if all peers have the same local parameter.

\subsection{Blockchain Basics}
\label{sec:blockchain:basics}
Blockchain, as a distributed ledger, refers to a chain of blocks linked by hashes and spread over all the  nodes in a peer-to-peer network, namely blockchain network. A full node stores a full blockchain in its local database. A blockchain starts from a genesis block, and each block except for the genesis block is chained to a previous block by referencing its hash. 
Typically, there are two categories of blockchain systems based on scale and openness: permissioned and permissionless. In this paper, we adopt a permissioned blockchain since it can provide faster speed and more restricted registration control than permissionless ones. 

Blockchain consists of three major components: blockchain network, distributed ledger, and consensus algorithm. A blockchain system organizes registered nodes into a P2P network, formulating a complete graph. A distributed ledger is immutable and can be organized as a chain, a Direct Acyclic Graph (DAG), or a mesh. In this paper, we use a chain as the data structure of our ledger. 
As the core of a blockchain system, the consensus process determines how to append a new block to the chain. Two types of consensus algorithms are commonly adopted: proof-of-resources and message passing. Proof-of-resources means that nodes compete for proposing blocks by demonstrating their utilization of resources, e.g., computational resources, stake, storage, memory and specific trust hardware. On the other hand, message passing based consensus has been widely researched in the area of distributed computing. Such algorithms always provide clear assumptions on nodes' faulty behaviors such as fail-stop and Byzantine attacks. In this paper, we leverage Byzantine fault tolerance (BFT) consensus algorithm, which can address Byzantine nodes who can launch arbitrary attacks.

\subsection{Gradient Aggregation Rule (GAR)}
A Gradient Aggregation Rule (GAR) is used to aggregate gradients received from peers during each round. A traditional GAR averages gradients to eliminate errors. Concerning gradients generated by Byzantine nodes, a GAR can be more elaborately designed to inject robustness. For example, Krum and Multi-Krum are the pioneering GARs that satisfy Byzantine resilience \cite{blanchard2017machine}. The essence behind these two GARs are to choose the gradient with the closest $(N-f)$ ($N$ denotes the network size and $f$ denotes the number of Byzantine nodes) neighbors based on the assumption that the honest majority should have similar gradients. Median \cite{yin2018byzantine} and MDA \cite{el2020genuinely} are another two GARs that adopt analogous ideas to ensure the BFT gradient aggregation. In this paper, SPDL leverages Krum as the BFT GAR but is not limited to it.

\subsection{Differential Privacy}
The privacy guarantee is of vital importance if the nodes carrying out decentralized learning do not admit their local training data to be shared. Although each node communicates with its neighbors by transmitting parameters instead of sending raw data, the risk of leaking information still exists \cite{wang2019beyond}.  Differential privacy is an effective method to avoid leaking any information of a single individual by adjusting the feedback of the query operations, no matter what auxiliary information the adversary node might have. In decentralized learning, the process of exchanging parameters involves a sequence of queries. The differential privacy in our setup can be formally defined as follows:
\begin{definition}
	(($\epsilon,\delta$)-differential-privacy) Denote by $\mathcal{A}$ an arbitrary synchronous distributed learning algorithm that updates model parameter $x$. For any node $i$ in a decentralized system and any two possible local data-sets $D_i$ and $D'_i$ with $D_i$ differing from $D'_i$ by at most one record, if for any round $t$ and any $S\subseteq Range(\mathcal{A})$, it holds that
	$$Pr\left(\mathcal{A}(\mathcal{D},D_i;t)\in S\right) \leq e^{\epsilon}Pr(\mathcal{A}(\mathcal{D},D'_i;t)\in S)+\delta,$$
	we then claim that $\mathcal{A}$ preserves ($\epsilon,\delta$)-differential-privacy, and
	pack all $D_j(j\neq i)$ into $\mathcal{D}$.
\end{definition}

\begin{definition} 
	For any function $f:\mathcal{D}\rightarrow R^{N}$, the $L_2$-sensitivity of $f$ is defined as
	$$\Delta_{2} f = \underset{d_{1},d_{2}}{\max} || f(d_{1})-f(d_{2}) ||,$$
	
	for all $d_{1},d_{2}$ differing in at most one element.
\end{definition}

Differential privacy can be realized by adding Gaussian noises to the query results \cite{Jiang2021TKDE}.  The following lemma, proved in \cite{yu2021decentralized}, demonstrates how to properly choose a Gaussian noise. 

\begin{lemma}\label{41}
	For each node transmitting gradients perturbed by Gaussian noise with distribution $\mathcal{N}(0,\sigma^{2})$, the gradient exchanges within $T$ successive rounds preserve ($\epsilon,\delta$)-differential-privacy as long as $\sigma\geq CT\gamma \sqrt{2\ln (1.25/\delta)}/\epsilon$, where $C=\Delta_{2} g^{(t)}$ and $\gamma$ is the learning rate.
\end{lemma}

In brief, the differential private scheme we employed in this paper is sketched as follows. For any node $i$, we use random Gaussian noise to perturb its gradients before transmitting to other nodes. When nodes obtain the result of the aggregation function $\mathcal{K}$ whose inputs are their perturbed gradients, the differential privacy for node $i$ is guaranteed.

\section{The Protocol}
\label{sec:protocol}

\subsection{Model and Assumptions}
We focus on scenarios where nodes are able to communicate with each other in a decentralized network. Specifically, we formalize the decentralized communication topology as a directed and fully connected graph $G=(V,E)$, where $V (|V|=N)$ denotes the set of all peers and for any $i$,$j\in V$,we have $(i,j)\in E$. Time is divided into epochs (denoted by $k$), with each consisting of  synchronous rounds (denoted by $t$), and a model can be trained within each epoch. 
We denote frequently-used notations of transaction, block, blockchain, chain of block headers, by $tx$, $B_{k}^{(t)}$, $BC_{k}^{(t)}$, and $BH_{k}^{(t)}$, respectively.

We assume the network is unreliable with at most $f$ possible Byzantine nodes, and $N=3f+1$ . A Byzantine node $i$ can behave arbitrarily. For example, it may refuse to compute gradient, transmit arbitrary but unlawful gradients to mislead correct nodes, and make incorrect votes during the consensus process. When $i$ chooses not to send any data in a synchronous round $t$, it's neighbor acts like receiving $g_i^{(t)}=0$.

\begin{figure*}[!htbp]
	\centering
	\includegraphics[width=7in]{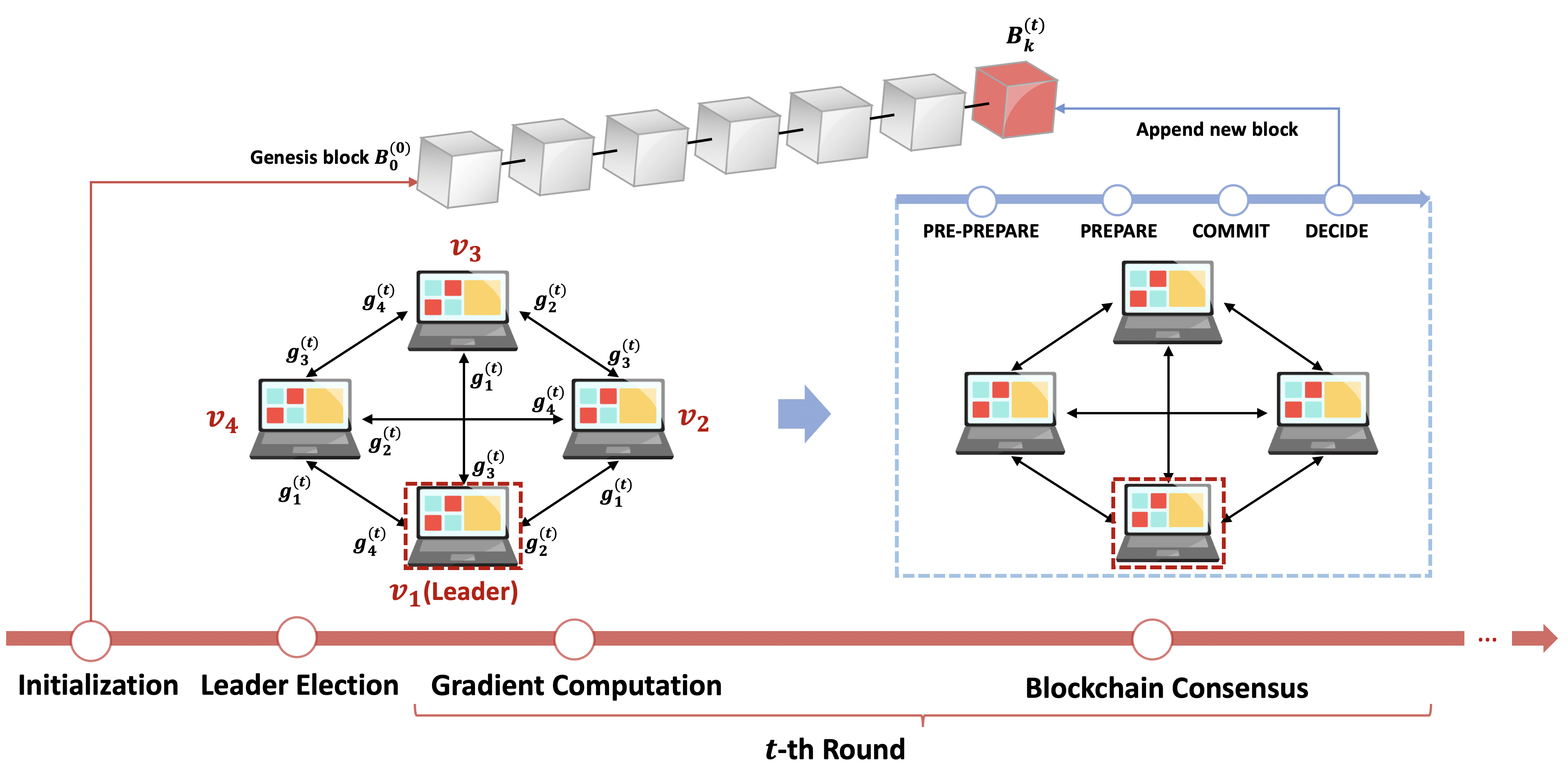}
	\caption{SPDL workflow ($t$-th round)}
	\label{fig:b-dpsdg}
\end{figure*}

\subsection{Design Objectives}
\label{sec:sub:design:objectives}
In this subsection, we briefly summarize our design goals.
\begin{enumerate}
\item Decentralization: SPDL should work in a decentralized network setting without the intervention of any centralized party such as a parameter server.  
\item Differential Privacy: SPDL should guarantee $(\epsilon,\delta)$-DP by adding random Gaussian noises to perturb gradients. Meanwhile we aim to reach a balance between privacy leakage and convergence rate.  
\item Byzantine Fault-Tolerance: SPDL can ensure convergence against at most $f(N=3f+1)$ Byzantine nodes who can behave arbitrarily.  
\item Immutability, Transparency, and Traceability: the full record of the machine learning  process should be immutable and transparent, and provide traceability enabling Byzantine node detection.
\end{enumerate}

\begin{table}
\caption{Summary of Notations}
\label{table:summary:notation}
\centering
\begin{tabular}{c|c}
  \hline
  \textbf{Symbol} & \textbf{Description} \\
  \hline
  	$\tilde{g}_{i}^{(t)}$ & the gradient computed by $i$ in round $t$\\
	$\mathcal{G}_{i}^{(t)}$ & the random noise added to $i$'s gradient in round $t$\\
	$g_i^{(t)}$ & the perturbed gradient to be transmitted in round $t$\\
	$\eta$ & the learning rate\\
	$x^{(t)}$ & model parameters at the end of round $t$\\
	$\xi_{i}^{(t)}$ & datasets randomly sampled from local datasets $D_i$\\
	$n, f$ & the number of nodes and Byzantine nodes\\
	$\sigma^2$ & the variance of Gaussian noise\\
	$\sigma_f^2$ & the upper bound of $\mathbb{E}\|\mathbb{E}(g_i^{(t)}) - g_i^{(t)}\|^2$\\
	$\Delta^{(t)}$ & the output of BFT GARs\\
	$B_{k}^{(t)}, BC_{k}^{(t)}$ & block and blockchain in $t$-th round of $k$-th epoch \\
	$\epsilon, \delta$ & budgets of differential privacy\\
	$d$ & the dimension of gradients\\
	$N(i)$ & the neighbors of node $i$\\
  \hline
\end{tabular}
\end{table}

\subsection{Protocol Details}

\subsubsection{Initialization} 
Initially, each node creates a pair of private key $\sk$ and public key $\pk$, and generates its unique 256-bit identity $id$ based on $\pk$. Public keys and identities are broadcast through the network to be publicly known by all nodes. Then a genesis block $B_0$ is created, which records the information of the nodes who initially participate in the blockchain network. A new coming node should have a related $tx$ being added to  the blockchain before joining the permissioned network, where $tx$ contains necessary information including its $pk$, $id$, IP address, etc. Initially, all nodes have an identical reputation value, that is $w_i=\hat{w}$.
%
Each node initializes itself with a learning rate $\gamma$, the number of total rounds $T$, and the variance of Gaussian noise for perturbing the gradients. For simplicity and consistency of the iterating process, we assume all nodes start the learning procedure with the same initial parameter value $x_i^{(0)}$. After initialization, nodes undergo a leader election process to determine who is in charge of the training process. 

\subsubsection{Leader Election}
\begin{algorithm}
	\DontPrintSemicolon
	\caption{Leader Election}
	\label{alg:leader:election}  
	$(h_i,\pi_i)$ = VRF($\sk_i$, $seed$)\;
	$L_i={(h_i,\pi_i, id_i)}$\;
	Broadcast $(h_i,\pi_i)$ to the network and receive $(h,\pi)$ form peers\;
	\While{TRUE}{
		\If {receive $(h_j, \pi_j)$ from $id_j$ \&\& VerifyVRF($pk, h_j, \pi_j , seed$)=1}{
			add $(h_j, \pi_j, id_j)$ to $L_i$\;
		}
		\If {$\mathtt{Time()}>start+\delta_1$}{
			select the largest $h_{max}$ from $L_i$ and obtain $id_{max}$\;
			\KwOut{$id_{max}$}
		}
	}
\end{algorithm} 

Each node executes the leader election algorithm shown in Algorithm~\ref{alg:leader:election}. The algorithm is based on the verifiable random function (VRF), which takes as inputs a private key $\sk$ and a random seed, and outputs a hash string $h$ as well as the corresponding proof $\pi$. Each contender broadcasts $(h_i,\pi_i)$ to the network and receives $(h,\pi)$ from its peers. Note that each node is assigned with a reputation variable $r\in [0,1]$. A node $i$ with $r_i=0$ is prohibited from being a leader. Then the node with the largest $h$ and $r>0$ is recognized as the leader who is responsible for the blockchain consensus. The case when more than one leaders are selected is extremely small since $h$ has a large space of $2^{256}$ if we adopt the commonly used SHA-256; but if this extreme case happens, all nodes relaunch the leader election process to ensure that only one leader is finally selected. We also set a timeout for the leader election process as $\mathtt{Time()}>start+\delta_1$, where $\mathtt{Time()}$ extracts the current UNIX time. To summarize, our leader election algorithm achieves the following three basic functionalities: 
\begin{itemize}
	\item A node with a zero reputation value has no right of being a leader.
	\item The leader election process possesses full randomness and unpredictability properties.
	\item A Byzantine node cannot disguise itself as a leader since VRF ensures that the proof $h$ is unforgeable.
\end{itemize}

After leader election, nodes start the round-based training process, with each round consisting of the gradient computation and blockchain consensus processes. 

\subsubsection{Gradient Computation} 
At each round, node $i$ exchanges its perturbed gradients with all other nodes in the blockchain network. Specifically, each node preserves a true stochastic gradient $\tilde{g}_{i}^{(t)}$ and a perturbed one $g_{i}^{(t)}$ to be shared. The whole exchange process can be summarized into the following steps:
\begin{itemize}
		\item \textbf{Local Gradient computation:} compute local stochastic gradient $\tilde{g}_{i}^{(t)} = \nabla F_{i}(x^{(t)},\xi_{i}^{(t)})$,  where $\xi_{i}^{(t)}$ is randomly sampled from local dataset $D_i$.
		\item \textbf{Adding noise:} add random Gaussian noise to the local gradient to be shared. The variance of the noise is denoted by input variable $\sigma$.
		\item \textbf{Broadcast gradients:} send the perturbed local gradients to all other nodes, and receive gradients from others at the same time.
\end{itemize}
	
\begin{algorithm}
	\DontPrintSemicolon
	\caption{Gradient Computation}
	\label{alg:dpsdg}  
	\textbf{Initialize:} $x_i^{(0)}$, learning rate $\gamma$, number of total rounds $T$, and variance of noise $\sigma$\;
	\For {$t=0$ to $T-1$}{
			$\triangleright$ \textcolor{blue}{Local Computation} \;
			Randomly sample $\xi_{i}^{(t)}$ and compute local stochastic gradient $\tilde{g}_{i}^{(t)}=\nabla F_{i}(x^{(t)},\xi_{i}^{(t)})$\;
			$\triangleright$ \textcolor{blue}{Adding Noise}\;
			Randomly generate Gaussian noise $\mathcal{G}_{i}^{(t)} \sim \mathcal{N}(0,\sigma^{2})$ and add noise to the variable $g_{i}^{(t)} = \tilde{g}_{i}^{(t)}+\mathcal{G}_{i}^{(t)}$\;
			$\triangleright$ \textcolor{blue}{Broadcast Gradients}\;
			Broadcast $g_{i}^{(t)}$ to the network and receive $g_{j}^{(t)}$ from each peer $j$\;
	}
\end{algorithm}

\subsubsection{Blockchain Consensus}
\begin{algorithm}
	\DontPrintSemicolon
	\caption{Blockchain Consensus}
	\label{alg:leader}  
	$\triangleright$ To prevent deadlock, each node starts a view change if $\mathtt{Time()}>start+\delta_2$\;
	$\triangleright$ \textcolor{BrickRed}{\texttt{PRE-PREPARE}}\;
	\If{role is leader}{
		$\Delta^{(t)}=\mathcal{K}(g_1^{(t)},g_2^{(t)},\cdots,g_n^{(t)})$\\
		$B_k^{(t)}\leftarrow \mathtt{MSGB}( \Delta^{(t)} )$\\
		broadcast $\langle\mathtt{PRE\raisebox{0mm}{-}PREPARE}, id, B_k^{(t)}, h\rangle_{\hat{\sigma}}$\\
	}
	$\triangleright$ \textcolor{BrickRed}{\texttt{PREPARE}}\;
	\If{role is follower}{
		compute $\tilde{\Delta}^{(t)}=\mathcal{K}(g_1^{(t)},g_2^{(t)},\cdots,g_n^{(t)})$\;
		\While{receive $\langle\mathtt{PRE\raisebox{0mm}{-}PREPARE}, id, B_k^{(t)}, h\rangle_{\hat{\sigma}}$}{
			\If{$\sigma$ and $B_k^{(t)}$ are valid and $\tilde{\Delta}^{(t)}\approx B_k^{(t)}.\Delta^{(t)}$}{
				broadcast $\langle\mathtt{PREPARE}, id, h, vote\rangle_{\hat{\sigma}}$\;
			}
		}
	}
	$\triangleright$ \textcolor{BrickRed}{\texttt{COMMIT}}\;
	\While{receive $2f+1$ $\langle\mathtt{PREPARE}, id, h, vote\rangle_{\hat{\sigma}}$}{
		broadcast $\langle\mathtt{COMMIT}, id, vote\rangle_{\hat{\sigma}}$\;
	}
	$\triangleright$ \textcolor{BrickRed}{\texttt{DECIDE}}\;
	\While{receive $2f+1$ $\langle\mathtt{COMMIT}, id, h, vote\rangle_{\hat{\sigma}}$}{
		$\mathtt{Append}(BC_k^{(t)}, B_k^{(t)})$\;
		$x_i^{(t+1)} = x_i^{(t)} - \gamma\Delta^{(t)}$\;
		Update reputation\;
	}
\end{algorithm}  

The blockchain consensus process deeply integrates a blockchain, a BFT consensus protocol (e.g., PBFT, Tendermint), and a BFT aggregation function (e.g., Krum, Median). In this paper, we adopt the Practical Byzantine Fault Tolerant (PBFT) protocol as our consensus backbone due to its effectiveness validated by the Hyperledger Sawtooth. The aggregation rule used in blockchain consensus is Krum. Concretely, the blockchain consensus consists of four phases: \texttt{PRE-PREPARE}, \texttt{PREPARE}, \texttt{COMMIT}, and \texttt{DECIDE}. 

In the \texttt{PRE-PREPARE} phase, the leader computes an aggregated gradient using an aggregation function $\Delta^{(t)}=\mathcal{K}(g_1^{(t)},g_2^{(t)},\cdots,g_n^{(t)})$, where $\mathcal{K}$ is a $(b,\alpha)$-Byzantine resilient Krum function. The core idea of Krum is to eliminate the gradients that are too far away from others. We use Euclid distance $\|g_{i}^{(t)} - g_{j}^{(t)}\|^2$ to measure how far two gradients are separated. Then we define $near(i)$ to be the set of $n-f-2$ closest gradients to $g_{i}^{(t)}$. We expect that the Krum function chooses one gradient $g_{i}^{(t)}$ that is the ``closest'' to its surrounding gradients. More precisely, the output of Krum is one of its input gradients, and the index of this gradient is:
$$arg \min_{i} \sum_{j\in near(i)} \|g_{i}^{(t)} - g_{j}^{(t)}\|.$$
$\mathtt{MSGB}$ takes as input $\Delta^{(t)}$ and forms a new block $B_k^{(t)}$ which records $\Delta^{(t)}$. Then the leader broadcasts a signed pre-prepare message as $\langle\mathtt{PRE\raisebox{0mm}{-}PREPARE}, id, B_k^{(t)}, h\rangle_{\hat{\sigma}}$.

Inthe \texttt{PREPARE} phase, each follower computes $\tilde{\Delta}^{(t)}=\mathcal{K}(g_1^{(t)},g_2^{(t)},\cdots,g_n^{(t)})$ based on its local perturbed gradients, then waits for pre-prepare messages. If a pre-prepare message is received, the follower first verifies the digital signature $\sigma$ and the block (height, block hash, etc.).  Then it compares $\tilde{\Delta}^{(t)}$ with $B_k^{(t)}.\Delta^{(t)}$. The requirement of $\tilde{\Delta}^{(t)}\approx B_k^{(t)}.\Delta^{(t)}$ means $\tilde{\Delta}^{(t)}\in (B_k^{(t)}.\Delta^{(t)}-\delta, B_k^{(t)}.\Delta^{(t)}+\delta)$, where $\delta$ is a small variation. This condition indicates that each follower should have a similar view on non-Byzantine gradients as the leader. If verification is passed, the follower broadcasts a prepare message $\langle\mathtt{PREPARE}, id, h, vote\rangle_{\hat{\sigma}}$. 

In the \texttt{COMMIT} phase, if a node receives $2f+1$ valid commit messages, it can broadcast a decision $\langle\mathtt{COMMIT}, id, vote\rangle_{\hat{\sigma}}$ and enter into the following $\mathtt{DECIDE}$ phase. 

In the \texttt{DECIDE} phase, upon receiving $2f+1$ valid commit messages, a node can append a new block $B_k^{(t)}$ to its local blockchain $BC_k$, update its local gradient as $x_i^{(t+1)} = x_i^{(t)} - \gamma\Delta^{(t)}$, and finally update its reputation. The reputation of a certain node $i$ can be reduced if its gradient deviates from the aggregated gradient by more than $\pi/2$. Even though we do not explicitly introduce the view change, we do have such process to address the case when a leader is a Byzantine node. To avoid the occurrence of a deadlock, we set a timeout in the blockchain consensus process. If $\mathtt{Time()}>start+\delta_2$, each node broadcasts a view change message $\langle \mathtt{VIEW\raisebox{0mm}{-}CHANGE}, id, h\rangle_{\sigma}$ and waits for other peers' responses. Upon receiving $2f+1$ view change messages, a node can abandon the current round.

\section{Theoretical Analysis}\label{sec:analysis}

In this section, we provide both convergence analysis and regret analysis on SPDL in the presence of Byzantine nodes. 

Without loss of generality, let $g_1^{(t)}, g_2^{(t)}, \cdots, g_{n-f}^{(t)}$ be the perturbed gradients sent out by honest nodes in round $t$, and $g_{n-f+1}^{(t)}, g_{n-f+2}^{(t)}, \cdots, g_n^{(t)}$ be the gradients sent out by possible Byzantine nodes in round $t$. We assume that the gradients derived by correct nodes are independently sampled from the random viable $\tilde{G}$ and that $\mathbb{E}(\tilde{G}) = g$. By adding Gaussian noise, a perturbed gradient then can be considered as an instance sampled from random variable $G=\tilde{G}+\mathcal{G}$, where $\mathcal{G}$ follows the Gaussian distribution of mean 0 and variance $\sigma^2$. Therefore we also have $\mathbb{E}(G)=g$. In this section, if we only concentrate on a certain round $t$, we  omit the superscript on variables when ambiguity can be avoided from context.

\begin{definition}
	($(k,f)$-Byzantine Resilience) Let $0<k\leq 1$ and $f$ be the number of Byzantine nodes in a distributed system, then our anti-Byzantine mechanism $\mathcal{K}$ is said to be $(k,f)$-Byzantine Resilient if
	\begin{equation*}
		h^{(t)}  =  \mathcal{K}(g_1^{(t)}, g_2^{(t)}, \cdots, g_n^{(t)})
	\end{equation*}
	satisfies that $\langle \mathbb{E}h^{(t)}, g^{(t)} \rangle \geq k\|g^{(t)}\|^2$.
\end{definition}
\begin{theorem}
	In round $t$, If $\|g\|^2f^{-\frac{3}{4}} > 18d\sigma_f^2$ and
	\begin{equation*}
		\epsilon > \frac{\sqrt{2C\ln (1.25/\delta)}}{C_2},
	\end{equation*}
	where
	\begin{equation}
		C_2 = (\frac{\|g\|^2}{18d}f^{-\frac{3}{4}} - \sigma_f^2)^{\frac{1}{2}},
	\end{equation}
	our anti-Byzantine mechanism $\mathcal{K}$  achieves $(k,f)$-Byzantine Resilience with
	\begin{equation*}
		k = 1-\frac{3\sqrt{2}f^{\frac{3}{2}}\sqrt{d}}{\|g\|}\left(\sigma_f^2+ \frac{2C^2\ln (1.25/\delta)}{\epsilon^2}\right).
	\end{equation*}
\end{theorem}
\begin{proof}
	Denote by $N(i)$ the set of $n-f-2$ closest gradients of the $i$-th node, $N_c(i)$ the collection of the perturbed gradients in $N(i)$ sent by correct nodes, while $N_f(i)$ the gradients in $N(i)$ sent by the Byzantine nodes. Let $i^*$ be the index of the gradient chosen by  $\mathcal{K}$, we then have
	\begin{equation*}
		\begin{aligned}
			\|\mathbb{E} h^{(t)} - g\|^2 &= \left\|\mathbb{E}\left( h^{(t)} - \frac{1}{|N_c(i^*)|}\sum_{j\in N_c(i^*)}\left(g_i+r_i \right)\right)\right\|^2\\
			&\leq \mathbb{E} \left\|\left( h^{(t)} - \frac{1}{|N_c(i^*)|}\sum_{j\in N_c(i^*)}\left(g_i+r_i \right)\right)\right\|^2\\
            \text{(the above} &\text{ inequality holds since $\|\cdot\|$ is convex.)}\\
			&\leq\sum_{i \notin \mathcal{B}} \mathbb{E} \left\|\left( g_i+r_i - \frac{1}{|N_c(i)|}\sum_{j\in N_c(i)}\left(g_j+r_j \right)\right)\right\|^2\\
			&+  \sum_{i \in \mathcal{B}} \mathbb{E} \left\|\left( \mathcal{B}_i - \frac{1}{|N_c(i)|}\sum_{j\in N_c(i)}\left(g_j+r_j \right)\right)\right\|^2
		\end{aligned}
	\end{equation*}
	If $i^* = i$ is one of the correct nodes, i.e., $i^*\notin \mathcal{B}$, 
	\begin{equation*}
		\begin{aligned}
			&\mathbb{E}\left\|\left( g_i+r_i - \frac{1}{|N_c(i)|}\sum_{j\in N_c(i)}\left(g_j+r_j \right)\right)\right\|^2\\
			&=\mathbb{E}\left\|\left(  \frac{1}{|N_c(i)|}\sum_{j\in N_c(i)}\left(g_i+r_i - g_j-r_j \right)\right)\right\|^2\\
			&=\frac{1}{|N_c(i)|^2}\mathbb{E}\left\|\left( \sum_{j\in N_c(i)}\left(g_i - g_j\right) + \left(r_i - r_j \right)\right)\right\|^2\\
			&=\frac{1}{|N_c(i)|^2}\mathbb{E}\left\|\left( \sum_{j\in N_c(i)}\left(g_i - g_j\right) + \left(r_i - r_j \right)\right)\right\|^2\\
			&\leq \frac{1}{|N_c(i)|} \sum_{j\in N_c(i)}\mathbb{E}\left\|g_i - g_j\right\|^2 + \mathbb{E}\left\|r_i - r_j \right\|^2\\
			&\leq 2d(\sigma^2+\sigma_f^2).
		\end{aligned}
	\end{equation*}

	The above inequality holds since $\mathbb{E}(r_i-r_j)$ and $\mathbb{E}(g_i-g_j)$ are both 0. Because there are exactly $n-f$ correct nodes, we have
	\begin{equation*}
		\begin{aligned}
			&\mathbb{E} \left\|\left( g_i+r_i - \frac{1}{|N_c(i)|}\sum_{j\in N_c(i)}\left(g_j+r_j \right)\right)\right\|^2\\
			&\leq 2d(n-f)(\sigma^2+\sigma_f^2).
		\end{aligned}
	\end{equation*}
	If $i^* = k$ is one of the Byzantine nodes, i.e., $i^*\in \mathcal{B}$, 
	\begin{equation*}
		\begin{aligned}
			&\mathbb{E} \left\|\left( \mathcal{B}_k - \frac{1}{|N_c(k)|}\sum_{j\in N_c(k)}\left(g_j+r_j \right)\right)\right\|^2\\
			&\leq \frac{1}{|N_c(k)|} \sum_{j\in N_c(k)}\mathbb{E} \|\mathcal{B}_k - g_j-r_j\|^2\\
			&\leq \frac{1}{|N_c(k)|} \sum_{j\in N_c(i)}\mathbb{E} \|g_i + r_i - g_j-r_j\|^2\\
			&+\frac{1}{|N_c(k)|} \sum_{j\in N_f(i)}\mathbb{E} \|g_i + r_i - g_j-r_j\|^2,\\
		\end{aligned}
	\end{equation*}
	where $i\notin \mathcal{B}$ is any correct node. By the definition of $N(i)$,  a node labeled $\zeta(i)$ is correct, but is farther away from any node in $N(i)$. Therefore we have:
	\begin{equation*}
		\begin{aligned}
			&\mathbb{E} \left\|\left( \mathcal{B}_k - \frac{1}{|N_c(k)|}\sum_{j\in N_c(k)}\left(g_j+r_j \right)\right)\right\|^2\\
			&\leq \frac{1}{|N_c(k)|} \sum_{j\in N_c(i)} \mathbb{E}\|g_i + r_i - g_j - r_j\|^2\\
			&+\frac{N_f(i)}{N_c(k)} \mathbb{E}\|g_i+r_i - g_{\zeta(i)} - r_{\zeta(i)}\|^2\\
			&\leq \frac{N_c(i)}{N_c(k)} 2d(\sigma^2+\sigma_f^2) + \frac{N_f(i)}{N_c(k)}  \sum_{j\notin \mathcal{B}\land i\neq i} \mathbb{E} \|g_i-g_j+r_i-r_j\|^2\\
			&\leq 2d(\sigma^2+\sigma_f^2)\left(\frac{N_c(i)}{N_c(k)}+ \frac{N_f(i)}{N_c(k)} (n-f-1)\right)\\
			&\leq 2d(\sigma^2+\sigma_f^2)\left(\frac{n-f-2}{n-2f-2}+\frac{b}{n-2f-2} (n-f-1)\right).\\
		\end{aligned}
	\end{equation*}
	Combining the two results where $i^*$ is a correct node or a Byzantine node, we have:
	\begin{equation*}
		\begin{aligned}
			\|\mathbb{E}h^{(t)} - g\|^2 &\leq 2d(\sigma^2+\sigma_f^2)\left(n-f+\frac{n-f-2}{n-2f-2}\right)\\
			&+2d(\sigma^2+\sigma_f^2)\left(\frac{b}{n-2f-2} (n-f-1)\right)\\
			&\leq 2d(\sigma^2+\sigma_f^2)(f+3+f(f-1)+f^2(f+2))\\
			&\leq 36df^{3}(\sigma^2+\sigma_f^2).
		\end{aligned}
	\end{equation*}
	Since 
	\begin{equation*}
		\epsilon > \frac{\sqrt{2C\ln (1.25/\delta)}}{C_2},
	\end{equation*}
	then 
	\begin{equation*}
		\epsilon^2\left(\frac{\|g\|^2}{18d}b^{-\frac{4}{3}} -\sigma_f^2\right) > 2c\ln (1.25/\delta).
	\end{equation*}
	Therefore
	\begin{equation*}
		\|g\| > 3\sqrt{2}b^{\frac{3}{2}}\sqrt{d}(\sigma_f^2+\sigma^2)^{\frac{1}{2}}.
	\end{equation*}
	Finally, we have
	\begin{equation*}
		\langle \mathbb{E} h^{(t)}, g \rangle \geq \left(\|g\| - 3\sqrt{2}b^{\frac{3}{2}}\sqrt{d}(\sigma_f^2+\sigma^2)^{\frac{1}{2}}\right) \|g\| = k\|g\|^2.
	\end{equation*}
\end{proof}

Regret analysis is commonly used in online learning to investigate the loss difference caused by two learning methods. Therefore we present a regret analysis on SPDL to show the  incurred loss difference with and without Byzantine nodes.

In a decentralized system with Byzantine nodes, let $A_t \in \mathbb{R}^{d\times 1}$ be the parameter model of a node in round $t$, and $X_t\in \mathbb{R}^{d\times n}$ be the raw data randomly sampled by nodes in round $t$. The loss function in round $t$ can be written as $F(A_t, X_t)$. Consider the destructive effect on the training process caused by Byzantine nodes, we denote by $\tilde{A}_t$ the parameter learned by the system per round in the absence of Byzantine nodes. It's meaningful to compare the gap of loss function computed by two different parameters $A_t$ and $\tilde{A}_t$. If there is no Byzantine node, we assume $\tilde{A}_t$ is updated by any correct gradient $g_{\xi(t)}^{(t)}$ given by a random node $\xi(t)\in \{1,2,\cdots, n\}$.
\begin{definition}
	(Regret)
	\begin{equation*}
		\begin{aligned}
			\mathcal{R}(A, \tilde{A}) = \sum_{i=0}^{T-1} F(\mathbb{E} A_t, X_t) - F(\mathbb{E} \tilde{A}_t, X_t).
		\end{aligned}
	\end{equation*}
\end{definition}
\begin{theorem}
	If the loss function satisfies $L_1$-Lipschitz continuity, and $L_1 < 1$, then $\mathcal{R}(A, \tilde{A}) \leq \rho \sqrt{T}+ o(\sqrt{T})$, where
	$$\rho = \frac{6L_1 f^{\frac{3}{2}}\sqrt{d(\sigma^2+\sigma_f^2)}}{1-L_1}.$$
\end{theorem}
\begin{proof}
	By the property of $L_1$-Lipschitz continuity, we have 
	\begin{equation*}
		\begin{aligned}
			&\sum_{i=0}^{T-1} F(\mathbb{E}A_t, X_t) - F(\mathbb{E}\tilde{A}_t, X_t)\\ 
			&\leq \sum_{i=0}^{T-1} L_1\|\mathbb{E}A_t - \mathbb{E}\tilde{A}_t\|\\
			&\leq  \sum_{i=1}^{T-1} L_1\|\mathbb{E}A_{t-1} - \mathbb{E}\tilde{A}_{t-1} +\eta \mathbb{E}\left(h^{(t-1)} - g_{\xi(t)}^{(t-1)}\right) \|\\
			&\leq  \sum_{i=1}^{T-1}L_1\mathbb{E}\|A_{t-1} - \tilde{A}_{t-1}\| +  \eta L_1\|\mathbb{E}h^{(t-1)} - \mathbb{E} g_{\xi(t-1)}^{(t-1)}\|.
		\end{aligned}
	\end{equation*}
	Since 
	\begin{equation*}
		\begin{aligned}
			&\mathbb{E}\|\mathbb{E} h^{(t)} - g^{(t-1)}\| \leq 6f^{\frac{3}{2}}\sqrt{d(\sigma^2+\sigma_f^2)},
		\end{aligned}
	\end{equation*}
	we have
	\begin{equation*}
		\begin{aligned}
			\eta \|\mathbb{E}h^{(t-1)} - \mathbb{E} g_{\xi(t-1)}^{(t-1)}\| \leq 6\eta f^{\frac{3}{2}}\sqrt{d(\sigma^2+\sigma_f^2)}.\\
		\end{aligned}
	\end{equation*}
	Then we obtain
	\begin{equation*}
		\begin{aligned}
			&\sum_{i=0}^{T-1} F(\mathbb{E}A_t, X_t) - F(\mathbb{E}\tilde{A}_t, X_t)\\ 
			&\leq \sum_{i=1}^{T-1}L_1\mathbb{E}\|A_{t-1} - \tilde{A}_{t-1}\| + 6L_1\eta f^{\frac{3}{2}}\sqrt{d(\sigma^2+\sigma_f^2)}\\
			&\leq {6L_1\eta f^{\frac{3}{2}}\sqrt{d(\sigma^2+\sigma_f^2)}} + \frac{6TL_1\eta f^{\frac{3}{2}}\sqrt{d(\sigma^2+\sigma_f^2)}}{1-L_1}.
		\end{aligned}
	\end{equation*}
\end{proof}
Thus the theorem can be immediately proved by setting $\eta = \displaystyle\frac{1}{\sqrt{T}}$.

\begin{figure*}[!htbp]
	\centering
	\subfigure[$N=4$]{
		\label{fig:sr}
		\centering
		\includegraphics[width=0.23\textwidth]{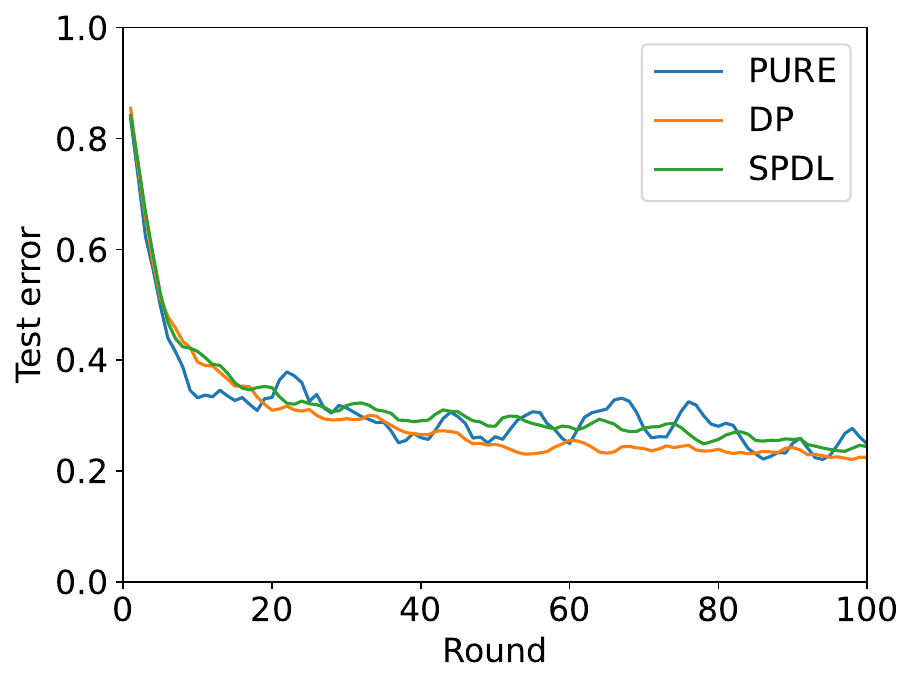}
	}
	\subfigure[$N=10$]{
		\label{fig:rw}
		\centering
		\includegraphics[width=0.23\textwidth]{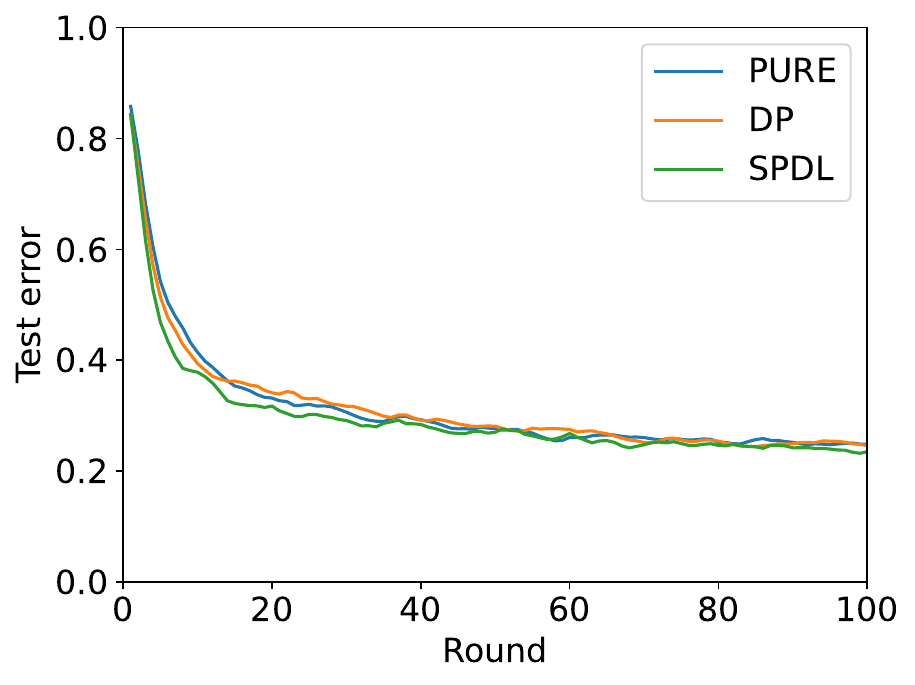}
	}
	\subfigure[$N=20$]{
		\label{fig:sr:dif}
		\centering
		\includegraphics[width=0.23\textwidth]{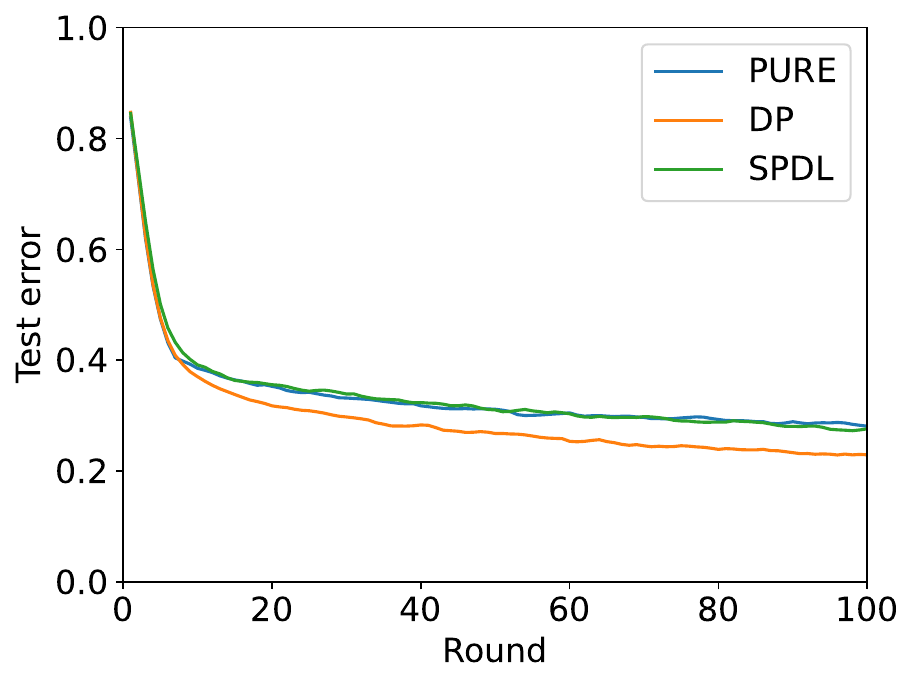}
	}
	\subfigure[$N=30$]{
		\label{fig:rw:dif}
		\centering
		\includegraphics[width=0.23\textwidth]{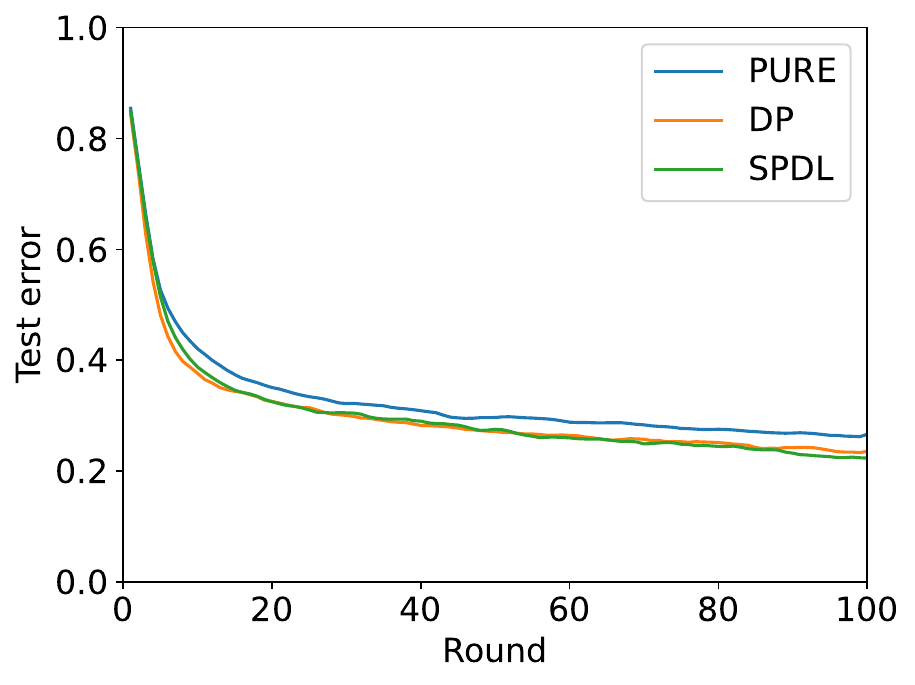}
	}
	\caption{Test error evolution with various network size $N=4, 10, 20, 30$.}
	\label{fig:convergence}
\end{figure*}
\begin{figure*}[!htbp]
	\centering
	\subfigure[$N=4$]{
		\label{fig:sr}
		\centering
		\includegraphics[width=0.23\textwidth]{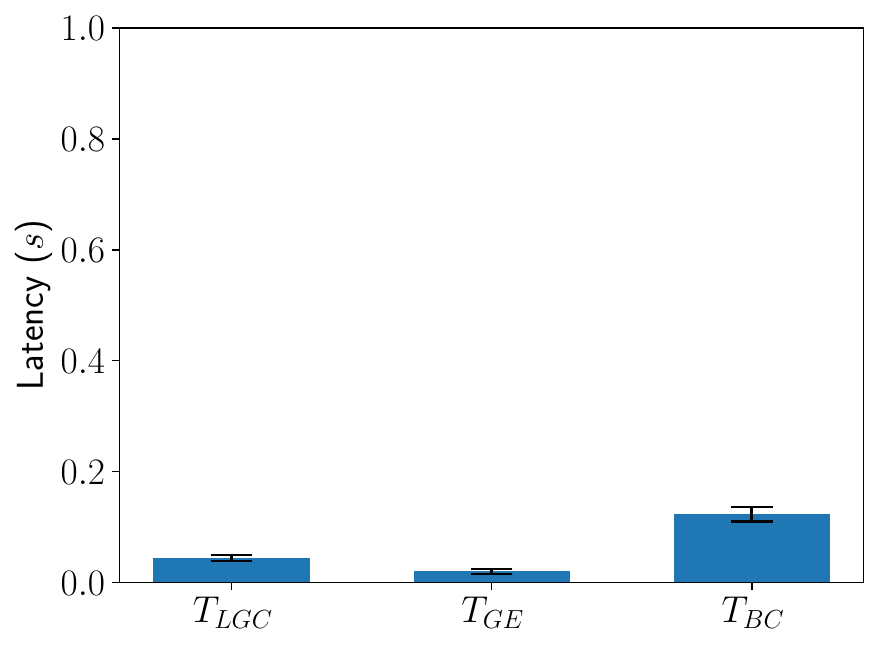}
	}
	\subfigure[$N=10$]{
		\label{fig:rw}
		\centering
		\includegraphics[width=0.23\textwidth]{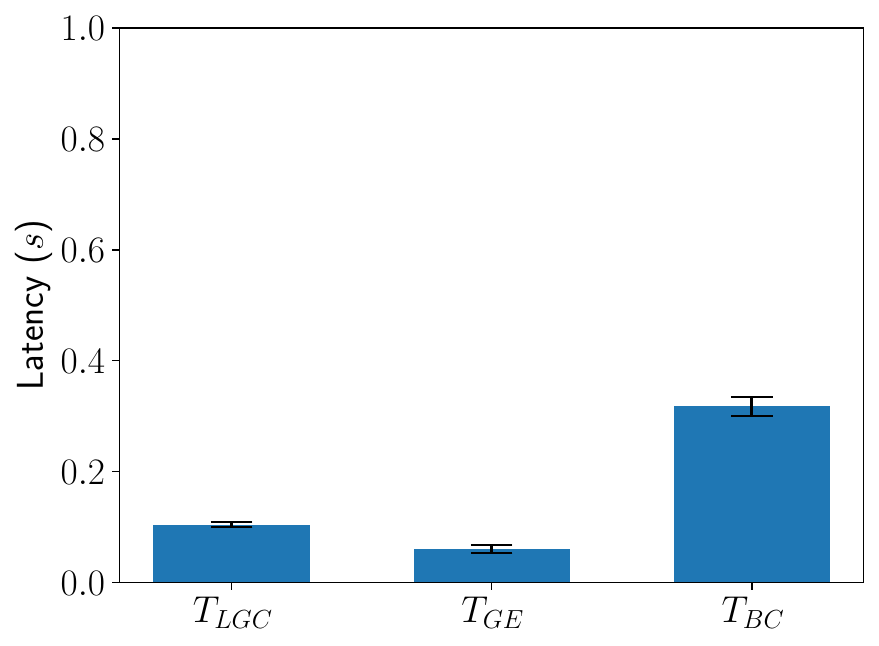}
	}
	\subfigure[$N=20$]{
		\label{fig:sr:dif}
		\centering
		\includegraphics[width=0.23\textwidth]{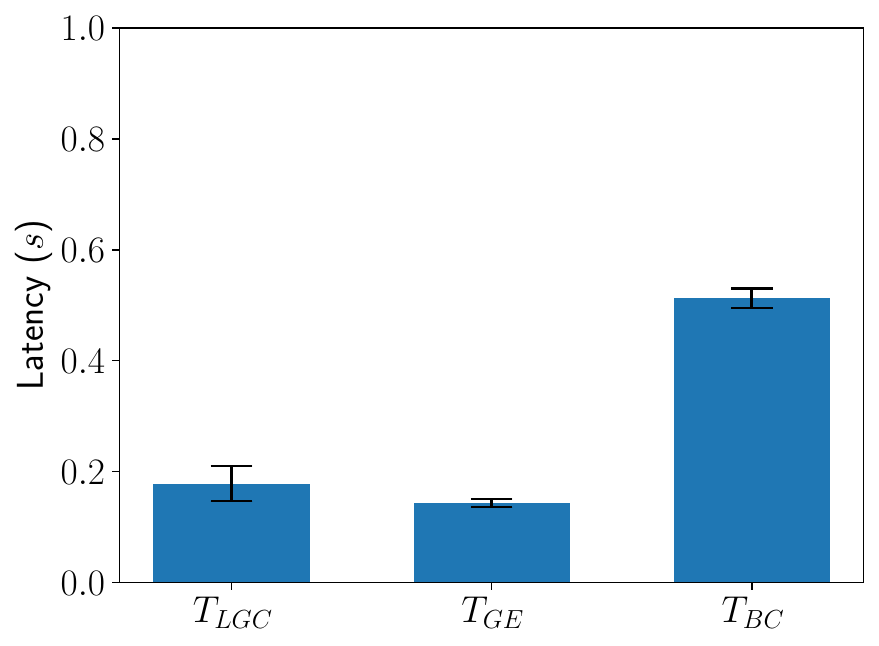}
	}
	\subfigure[$N=30$]{
		\label{fig:rw:dif}
		\centering
		\includegraphics[width=0.23\textwidth]{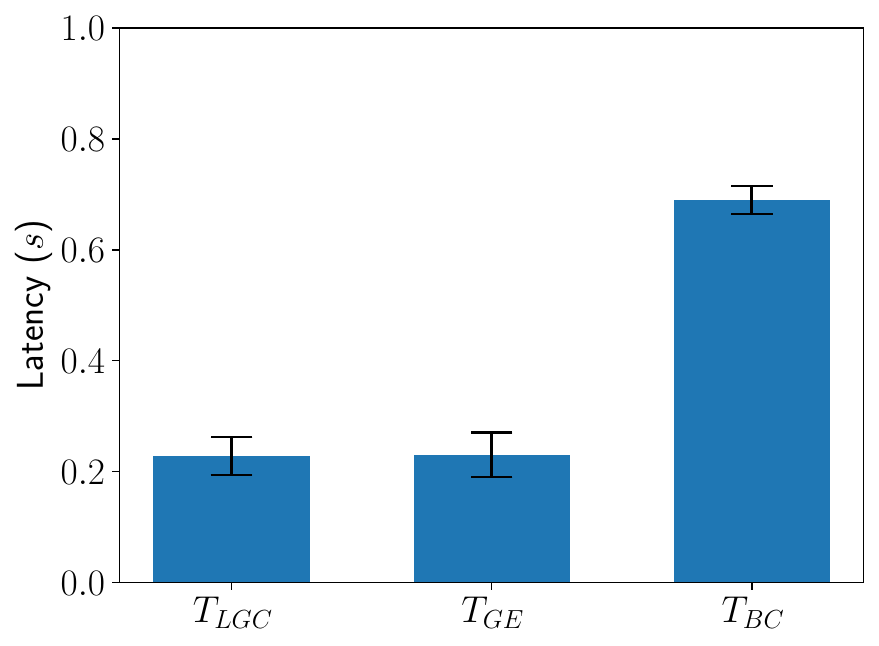}
	}
	\caption{Latency of different stages ($T_{LGC}$, $T_{GE}$ and $T_{BC}$) concerning $N=4, 10, 20, 30$.}
	\label{fig:latency}
\end{figure*}

\begin{figure*}[!htbp]
	\centering
	\subfigure[$BR=0\%$]{
		\label{fig:Byzantine:br0}
		\centering
		\includegraphics[width=0.23\textwidth]{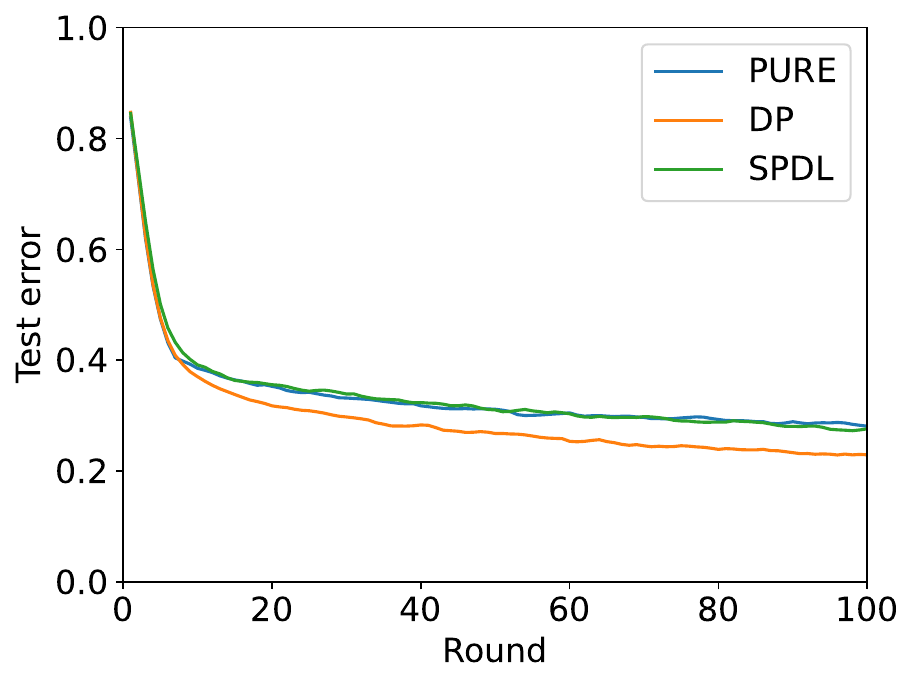}
	}
	\subfigure[$BR=10\%$]{
		\label{fig:Byzantine:br10}
		\centering
		\includegraphics[width=0.23\textwidth]{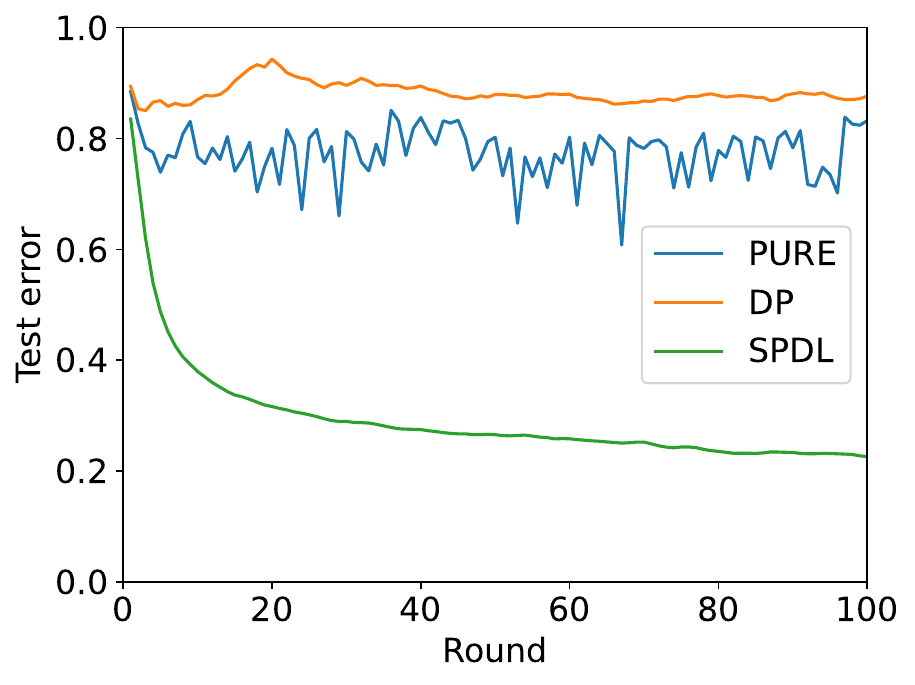}
	}
	\subfigure[$BR=20\%$]{
		\label{fig:Byzantine:br20}
		\centering
		\includegraphics[width=0.23\textwidth]{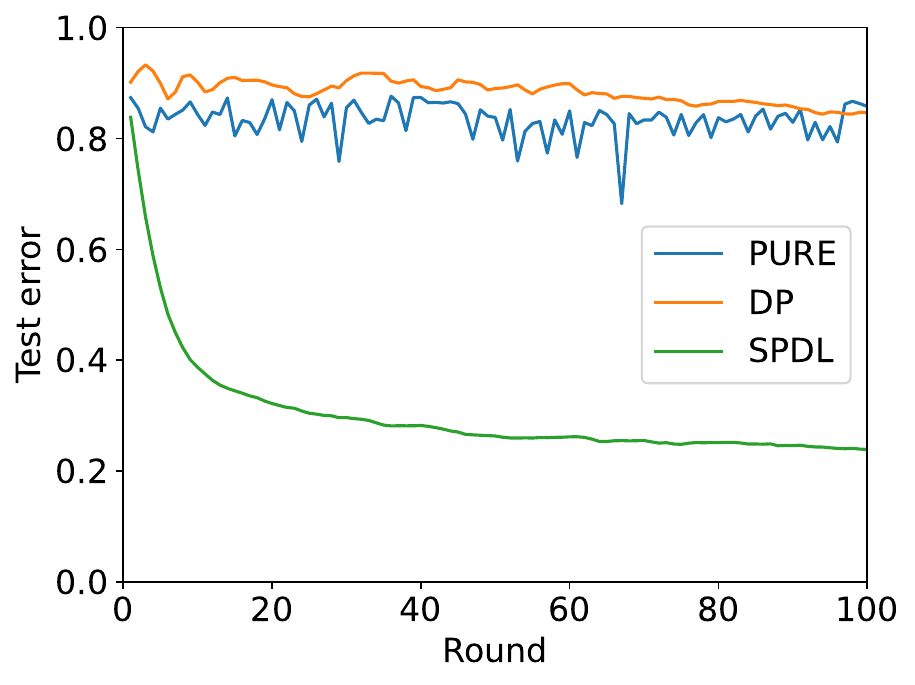}
	}
	\subfigure[$BR=30\%$]{
		\label{fig:Byzantine:br30}
		\centering
		\includegraphics[width=0.23\textwidth]{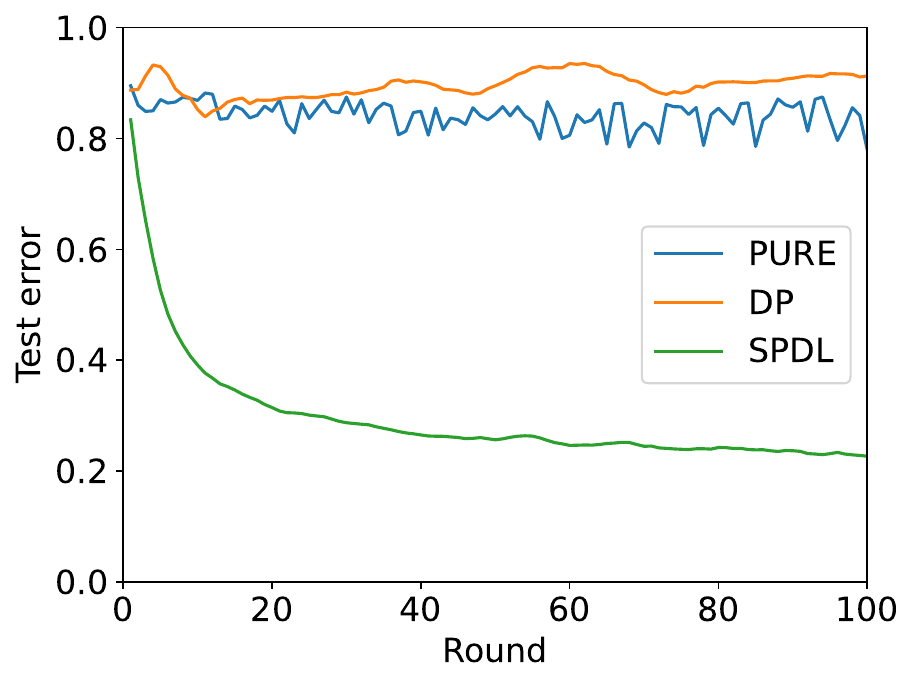}
	}
	\caption{Test error evolution with various Byzantine ratio $BR\in{0\%, 10\%, 20\%, 30\%}$.}
	\label{fig:Byzantine}
\end{figure*}

\begin{figure}[!htbp]
	\centering
	\subfigure[$BS = 10$]{
		\label{fig:batch:size:BS10}
		\centering
		\includegraphics[width=0.22\textwidth]{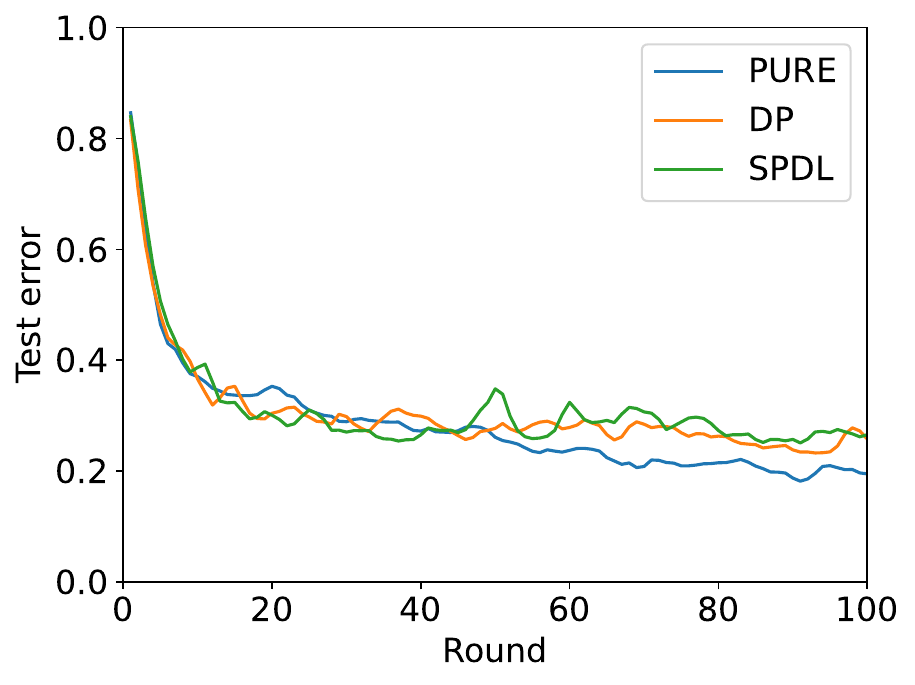}
	}
	\subfigure[$BS = 100$]{
		\label{fig:batch:size:BS100}
		\centering
		\includegraphics[width=0.22\textwidth]{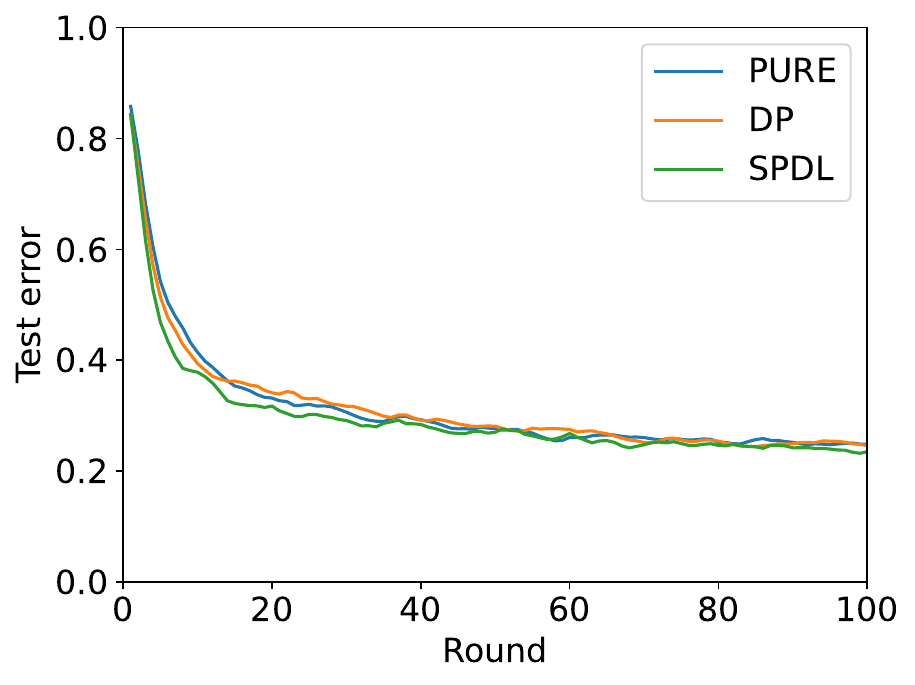}
	}
	\caption{Test error evolution with batch sizes $BS=10, 100$.}
	\label{fig:batch:size}
\end{figure}
\begin{figure}[!htbp]
	\centering
	\subfigure[$N = 10$]{
		\label{fig:epsilon:N10}
		\centering
		\includegraphics[width=0.22\textwidth]{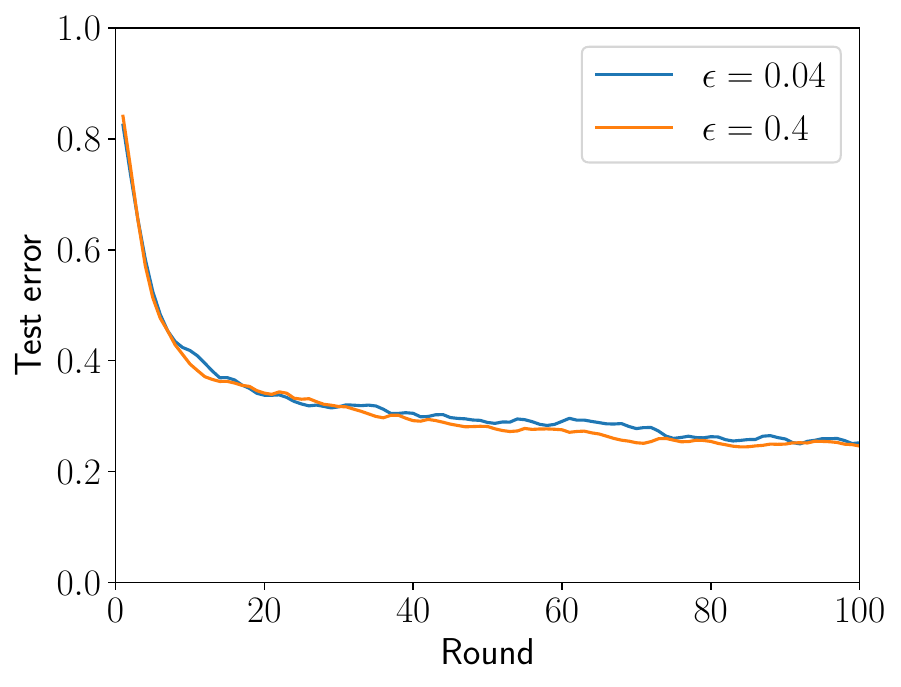}
	}
	\subfigure[$N = 20$]{
		\label{fig:epsilon:N20}
		\centering
		\includegraphics[width=0.22\textwidth]{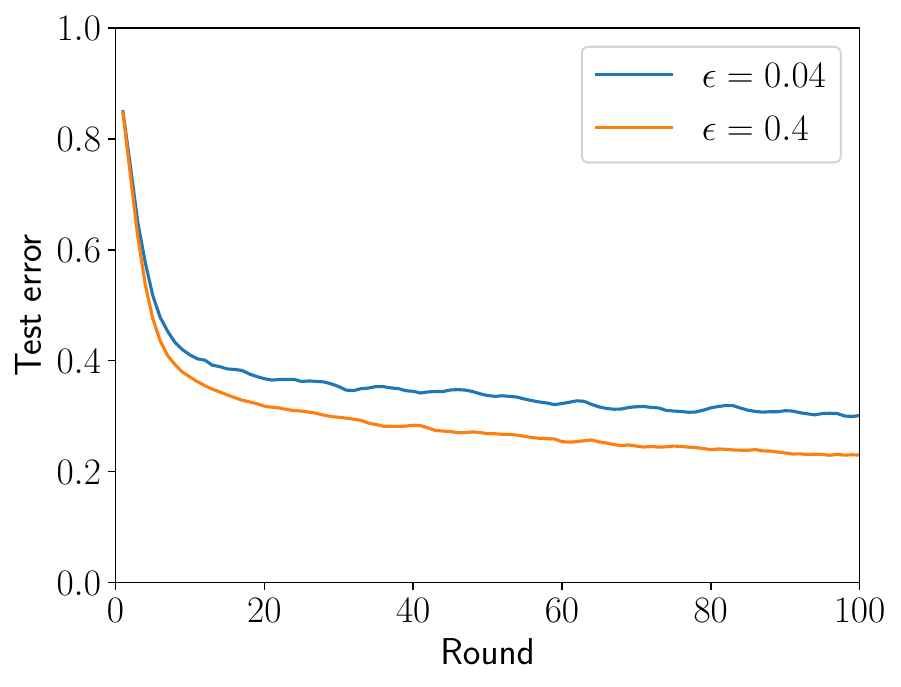}
	}
	\caption{Test error evolution with $\epsilon=0.4, 0.04$ and $N=10, 20$.}
	\label{fig:epsilon}
\end{figure}

\section{Evaluation} \label{sec:evaluation}

\subsection{Configuration}
We implement SPDL with 3500 lines of Python code and conduct the experiments on a DELL PowerEdge R740 server which has 2 CPUs (Intel Xeon 4214R) with 24 cores (2.40 GHz) and 128 GB of RAM. SPDL adopts the gRPC framework, a P2P network for underlying communications, Pytorch for machine learning libraries, and a blockchain system with the PBFT consensus algorithm. We make SPDL open-sourced at Github\footnote{https://github.com/isSPDL/SPDL}. Nodes bootstrap by generating key pairs using ECDSA, initializing the genesis block, establishing the gRPC connection, exchanging node list, and joining the P2P network. We evaluate SPDL over the image classification of MNIST dataset, which consists of handwritten digits of 70,000 $28\times 28$ images in 10 classes. The dataset is equally divided into $N$ groups, with each assigned to one node. Each node can add Gaussian noise to its local gradients with the setting of $\epsilon=0.02$ (if not stated otherwise) and $\delta=10^{-6}$. We evaluate the performance of SPDL using the following standard metrics. 1) Test error: the fraction of wrong predictions among all predictions, using the test dataset. We measure the test error with respect to rounds, network size, batch size, privacy budget, and Byzantine ratio. 2) Latency: the latency of each round. 

\subsection{Evaluation Results}

\textbf{Convergence with Network Size:} For simplicity in our context, we denote ``PURE'' as the decentralized learning scheme without leveraging any DP technique, BFT GARs, and blockchain system, and use ``DP'' to represent a decentralized learning scheme using the DP technique only based on ``PURE''. We first compare our SPDL with PURE and DP schemes in a non-Byzantine environment. As shown in Fig.~\ref{fig:convergence}, the test error nearly converges after 20 rounds, but fluctuates a lot when $N$ is as small as four. When $N=30$, all schemes almost achieve the same convergence. The SPDL and DP schemes sometimes (e.g., $N=20$ or $N=30$ in our experiments) have lower test error than PURE because adding noises could prevent the training process from over-fitting. Besides, the network size does not impact the convergence rate and a large network size contributes to stable convergence. 

\textbf{Latency:} To better illustrate the latency of each round, we divide a round into three stages: local gradient computation plus adding noise whose overall time overhead is denoted by $T_{LGC}$, gradient exchange ($T_{GE}$), and blockchain consensus ($T_{BC}$). As Fig.~\ref{fig:latency} shows, $T_{LGC}$, $T_{GE}$ and $T_{BC}$ are in the same order of magnitude. $T_{GE}$ grows with $N$ simply because more nodes contend for computational resources. The MNIST classification task can be finished quickly ($<$0.1 s/round), so $T_{LGC}$ is lower than $T_{BC}$ in our experiments. When the machine learning task becomes more difficult (e.g., 10 min/round), $T_{BC}\approx 1 s$ is an acceptable overhead and could be even ignored. 

\textbf{Convergence in the Presence of Byzantine Nodes:} We then make a comparison of three different schemes {PURE, DP, SPDL} with respect to Byzantine Ratio ($BR$), which is the number of existing Byzantine nodes over $N$. We set $N=20$ and $BR\in{0\%, 10\%, 20\%, 30\%}$ considering $f=33\%\times N$, where $BR=0$ represents the non-Byzantine case. As shown in Fig.~\ref{fig:Byzantine}, the results of the non–Byzantine experiments indicate that the three schemes can achieve similar convergence. However, DP and PURE schemes have high test error when $BR>0\%$, and fail to ensure model convergence even in the presence of $10\%N$ Byzantine nodes. It is clearly shown that SPDL can still grantee the same convergence with respect to different levels of Byzantine attacks.

\textbf{Batch Size:} We then present the performance of the three schemes {PURE, DP, SPDL} with two different batch sizes (abbreviated as ``BS'') in Fig.~\ref{fig:batch:size}. When $BS=10$, the test error in all deployments fluctuate a lot, with the PURE scheme outperforming others because adding noises can perturb the model convergence. However, Fig.~\ref{fig:batch:size:BS100} indicates that we can increase the batch size to ensure a stable convergence and make our SPDL perform well as a PURE scheme.

\textbf{Privacy Budget:} We finally test our SPDL scheme by setting $\delta=10^{-6}$ with two varying $\epsilon\in{0.4, 0.04}$. A smaller $\epsilon$ represents stronger privacy guarantee. The results presented in Fig.~\ref{fig:epsilon} demonstrate that when $N=10$, adding noise with $\epsilon=0.4$ or $\epsilon=0.04$ have similar convergence. However, when $N=20$, smaller $\epsilon$ can cause larger test error. This implies that the tradeoff between accuracy and privacy preservation should be carefully adjusted according to specific demands on privacy protection and model accuracy.

\section{Conclusion} \label{sec:conclusion}
SPDL is a new decentralized machine learning scheme which ensures efficiency while achieving strong security and privacy gurantee. In particular, SPDL utilizes BFT consensus and BFT GAR to protect model updates from harsh Byzantine behaviors, leverages blockchain to enjoy the benefits of transparency and traceability, and adopts the DP technique for privacy protection. We provide rigorous theoretical analysis on the effectiveness of our scheme and conduct extensive studies on the performance of SPDL with variations of network size, batch size, privacy budget, and Byzantine ratio.

\bibliographystyle{IEEEtran}
\bibliography{ref}

\begin{IEEEbiography}[{\includegraphics[width=1in,height=1.25in,clip,keepaspectratio]{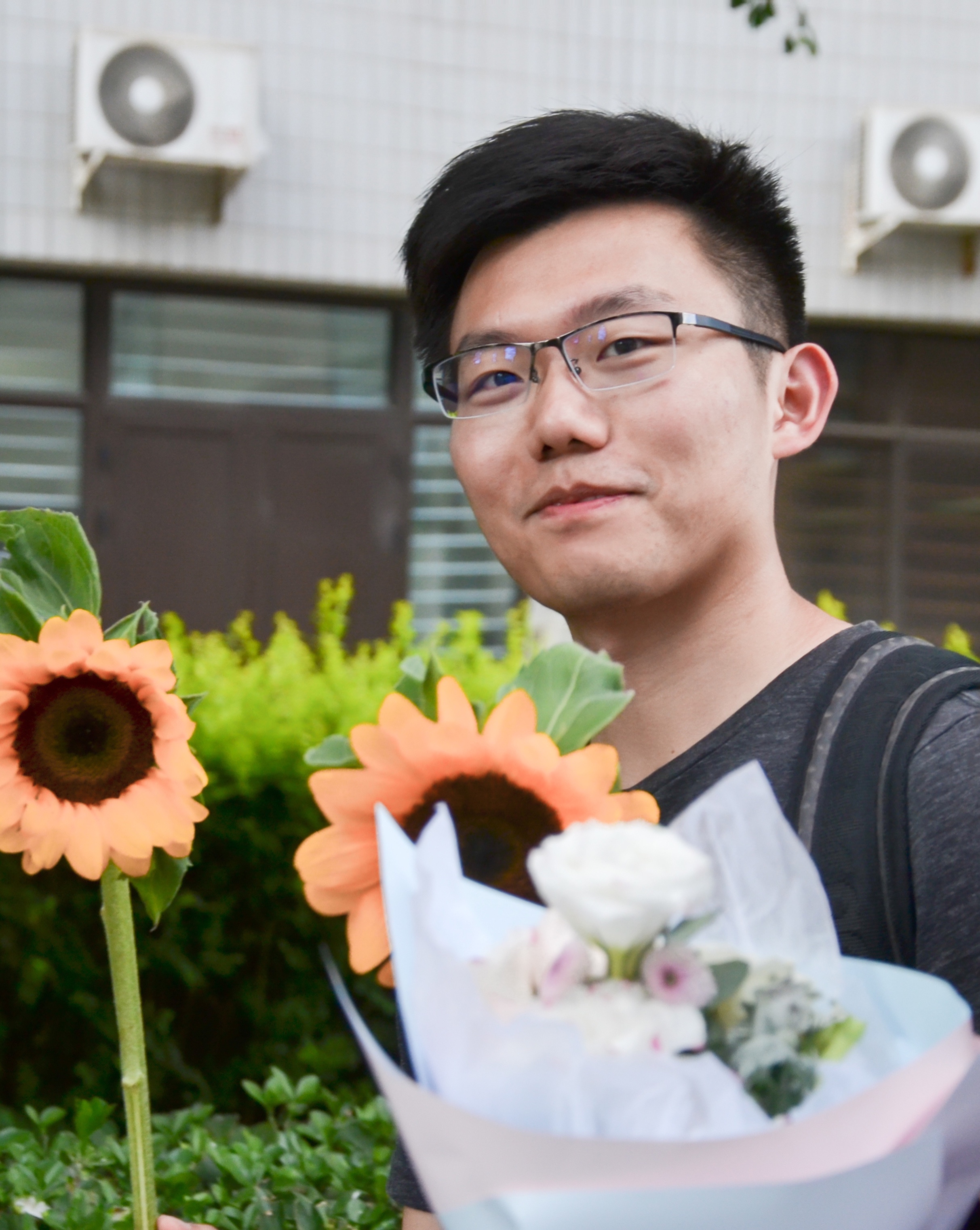}}]{Minghui Xu} received his PhD degree in Computer Science from The George Washington University in 2021, and received the BS degree in Physics from the Beijing Normal University in 2018. He is currently an Assistant Professor
in the School of Computer Science and Technology, Shandong University, China. His current research focuses on blockchain, distributed computing, and applied cryptography.
\end{IEEEbiography}

\begin{IEEEbiography}[{\includegraphics[width=1in,height=1.25in,clip,keepaspectratio]{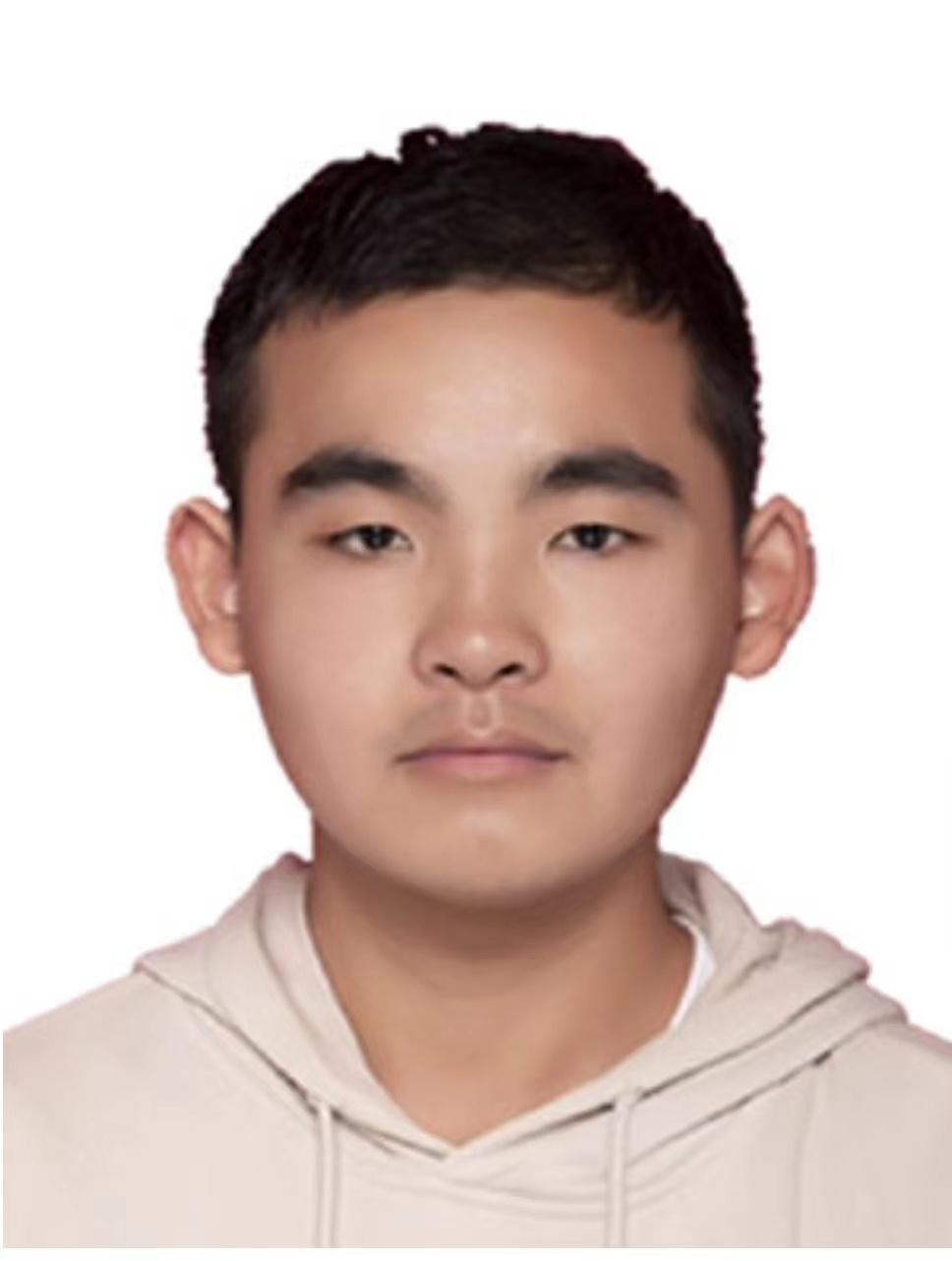}}]{Zongrui Zou} is currently working toward the under-graduate degree with the School of Computer Science and Technology, Shandong University, Qingdao, China. His research interests mainly include theoretical aspects of private data analysis and machine learning.
\end{IEEEbiography}

\begin{IEEEbiography}[{\includegraphics[width=1in,height=1.25in,clip,keepaspectratio]{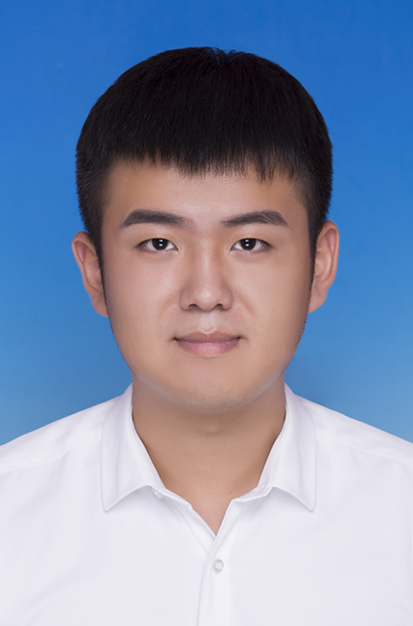}}]{Ye Cheng} received his bachelor's degree in mechanical engineering from Wuhan University of Technology in 2018. He is working toward a master’s degree in Computer Science and Technology at Shandong University in China. His current research direction is blockchain and privacy protection.
\end{IEEEbiography}

\begin{IEEEbiography}[{\includegraphics[width=1in,height=1.25in,clip,keepaspectratio]{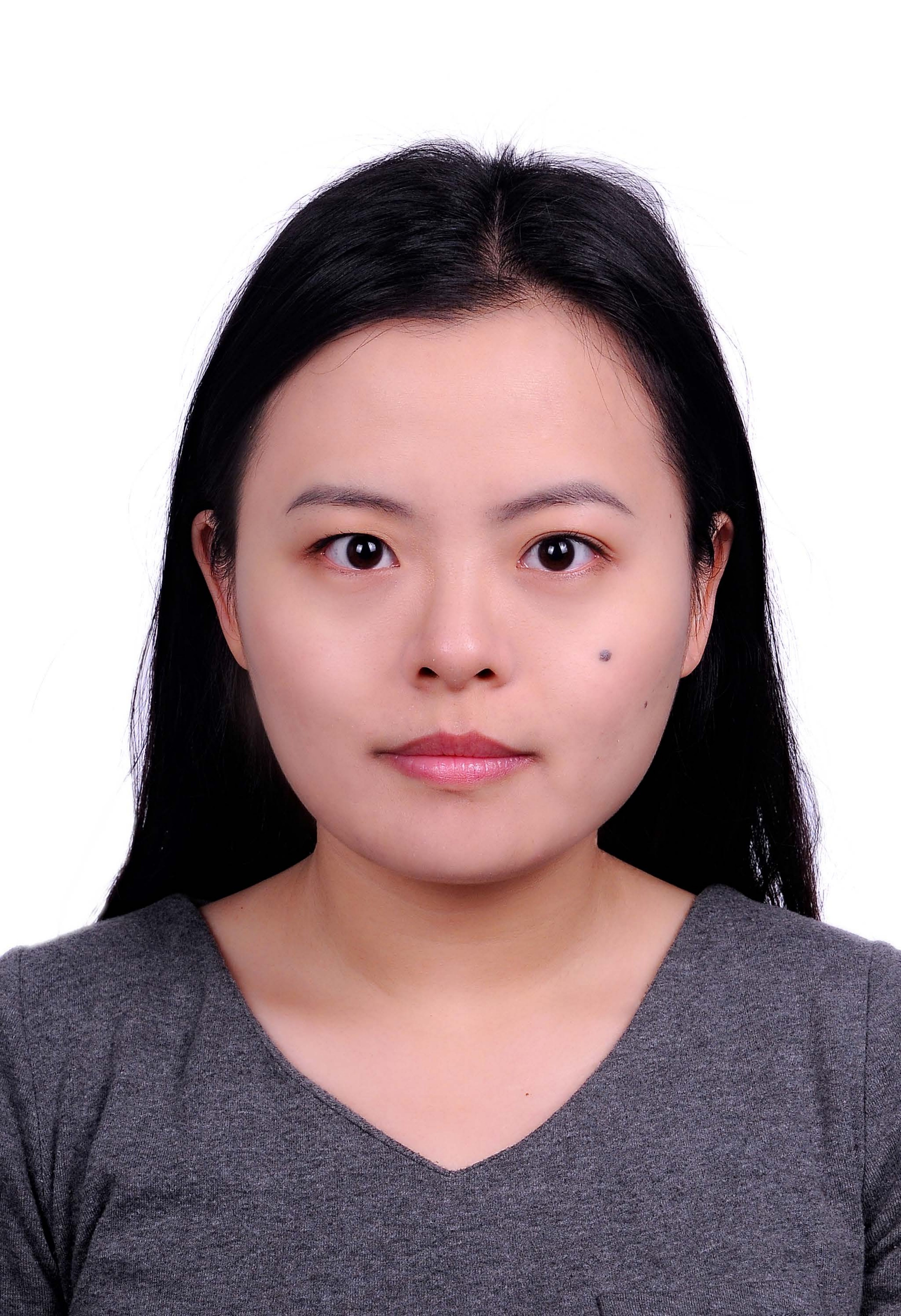}}]{Qin Hu} received her Ph.D. degree in Computer Science from the George Washington University in 2019. She is currently an Assistant Professor with the Department of Computer and Information Science, Indiana University-Purdue University Indianapolis (IUPUI). Her research interests include wireless and mobile security, edge computing, blockchain, and crowdsourcing/crowdsensing.
\end{IEEEbiography}

\begin{IEEEbiography}[{\includegraphics[width=1in,height=1.25in,clip,keepaspectratio]{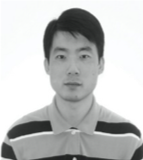}}]{Dongxiao Yu} received his BS degree in Mathematics in 2006 from Shandong University, and PhD degree in Computer Science in 2014 from The University of Hong Kong. He became an associate professor in the School of Computer Science and Technology, Huazhong University of Science and Technology, in 2016. Currently he is a professor at the School of Computer Science and Technology, Shandong University. His research interests include wireless networking, distributed computing, and graph algorithms.
\end{IEEEbiography}

\begin{IEEEbiography}[{\includegraphics[width=1in,height=1.25in,clip,keepaspectratio]{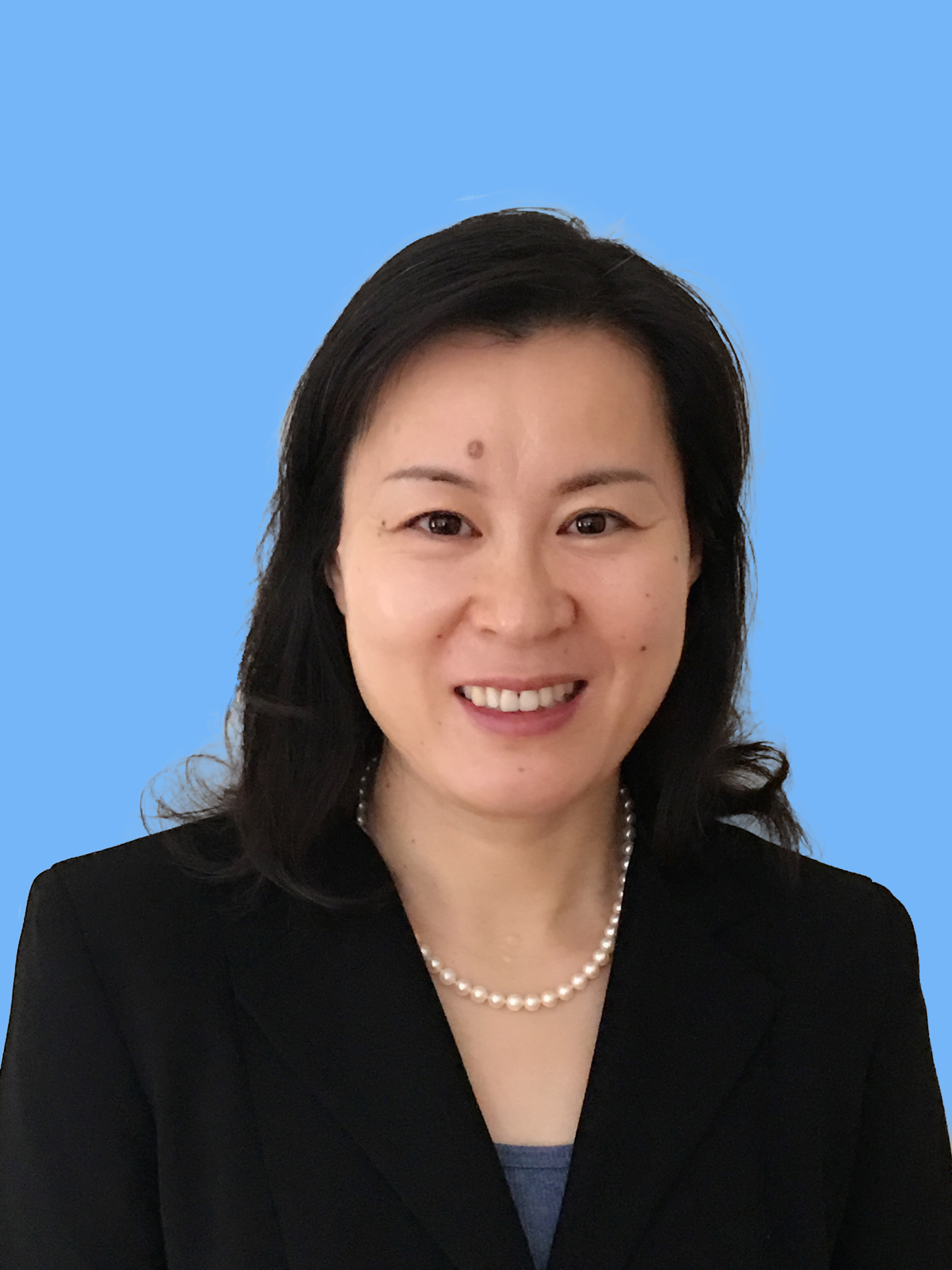}}]{Xiuzhen Cheng} received her MS and PhD degrees in computer science from University of Minnesota, Twin Cities, in 2000 and 2002, respectively. She was a faculty member at the Department of Computer Science, The George Washington University,  from 2002-2020. Currently she is a professor of computer science at Shandong University, Qingdao, China. Her research focuses on blockchain computing, security and privacy, and Internet of Things. She is a Fellow of IEEE.
\end{IEEEbiography}

\end{document}